\documentclass[a4paper,11pt,reqno,final]{amsart}

\usepackage{amsmath}
\usepackage{amssymb}
\usepackage{amsthm}
\usepackage{graphicx}
\usepackage{dcolumn}
\usepackage{bm}
\usepackage{psfrag}      
\usepackage{subfigure}
\usepackage{pstricks}
\usepackage{setspace}  

\numberwithin{equation}{section}

\newtheorem{theorem}{Theorem}[section]
\newtheorem{lemma}[theorem]{Lemma}
\newtheorem{proposition}[theorem]{Proposition}
\newtheorem{corollary}[theorem]{Corollary}

\newenvironment{remark}[1][Remark]{\begin{trivlist}
\item[\hskip \labelsep {\bfseries #1}]}{\end{trivlist}}

\def\Dcal{\mathcal{D}}
\def\Ecal{\mathcal{E}}
\def\Fcal{\mathcal{F}}
\def\R{{\mathbb R}}
\def\Z{\mathbb{Z}}
\def\N{\mathbb{N}}

\newcommand{\PiD}{\Pi_{\Dcal}}
\newcommand{\Dconv}{\mathcal{C}}
\newcommand{\PiA}{\Pi_{\Dconv}}
\def\Id{{\mathrm Id}}

\DeclareMathAlphabet{\bi}{OML}{cmm}{b}{it}
\DeclareMathAlphabet{\bcal}{OMS}{cmsy}{b}{n}

\begin{document}

\title[The dynamics of bargaining]{A projected gradient dynamical system modeling the dynamics of bargaining}

\author[D. Pinheiro]{D. Pinheiro}
\address[D. Pinheiro]{CEMAPRE, ISEG, Universidade T\'ecnica de Lisboa, Lisboa, Portugal}
\email{dpinheiro@iseg.utl.pt}
\author[A. A. Pinto]{A. A. Pinto}
\address[A. A. Pinto]{LIAAD-INESC Porto LA and Dep of Mathematics, Faculty of Science, University of Porto, Portugal}
\email{aapinto@math.uminho.pt}
\author[S. Z. Xanthopoulos]{S. Z. Xanthopoulos}
\address[S. Z. Xanthopoulos]{University of the Aegean, Samos, Greece}
\email{sxantho@aegean.gr}
\author[A. N. Yannacopoulos]{A. N. Yannacopoulos}
\address[A. N. Yannacopoulos]{Athens University of Economics and Business, Athens, Greece}
\email{ayannaco@aueb.gr}
\date{}

\begin{abstract}
We propose a projected gradient dynamical system as a model for a bargaining scheme for an asset for which the two interested agents have personal valuations which do not initially coincide. The personal valuations are formed using subjective beliefs concerning the future states of the world and the reservation prices are calculated using expected utility theory. The agents are not rigid concerning their subjective probabilities and are willing to update them under the pressure to reach finally an agreement concerning the asset. The proposed projected dynamical system, on the space of probability measures, provides a model for the evolution of the agents beliefs during the bargaining period and is constructed so that agreement is reached under the minimum possible deviation of both agents from their initial beliefs. The convergence results are shown using techniques from convex dynamics and Lyapunov function theory. 


\end{abstract}

\maketitle

\section{Introduction}
Projected dynamical systems  have attracted a lot of attention from the dynamical systems community recently, as they present a lot of interesting features both from the theoretical point of view, as well as from the point of view of applications. Such applications are mainly related to optimization theory in the form of gradient schemes for the solution of convex optimization problems. However, there are also important related applications to variational inequalities, general equilibrium theory, population dynamics and mechanics, among others. Most of the theory has been developed for continuous time dynamical systems \cite{Xia_Wang,Xia}. The aim of the present paper is twofold: (a) to investigate a novel application of projected dynamical systems into the theory of bargaining, and using that as a motivation (b) to investigate the theory of projected dynamical systems in the case of discrete-time dynamics. We believe that the application is interesting and wide enough to be studied in its own right. On the other hand, the results in this paper concerning stability and convergence of projected dynamical systems are general enough to apply to a wide range of similar dynamical systems arising in different applications or as numerical schemes for continuous time projected dynamical systems, and thus in our opinion worthy of reporting.

Bargaining theory is a branch of game theory that plays an important r\^ole in various applications in economic theory. Quoting Muthoo \cite{Muthoo} ``a bargaining situation is a situation in which two players\footnote{A ``player'' can be either an individual, or an organization (such as a firm, or a country).} have a common interest to co-operate. To put it differently, the players can mutually benefit from reaching agreement on an outcome from a set of possible outcomes (that contains two or more elements), but have conflicting interests over the set of outcomes.'' This is evidently a rather general description, under which a great variety of situations fall, and it is quite understandable that bargaining theory has become rather popular as a tool for economic modeling. To mention just a few applications, bargaining theory has found applications in diverse fields such as modeling of trade union negotiations, negotiations between countries, optimal asset ownership in firms, moral hazard in teams, bilateral monopoly, market games, among others, see e.g. \cite{Muthoo} and the extensive list of references both from the point of view of theoretical development of the subjects as well as from the point of view of applications.

One of the first seminal contributions in the field came from the work of Nash \cite{Nash} who proposed an axiomatic approach to the problem of bargaining based on the maximization of the product of utility functions of the two players (agents) involved.  His approach was a static one, but opened the road for the embedding of bargaining problems within the body of game theory, and provided important insights on the processes and variables that play on important r\^ole on the agents' decisions. His work made an impact, and generated a lot of important scientific research on the field, as well as alternative axiomatic approaches (see e.g. \cite{Kalai}). Another seminal contribution was Rubinstein's model \cite{Rubinstein} which undertook a more realistic approach towards the bargaining problem. This was an alternating offers model, according to which the two agents were considering alternate offers in splitting an asset comparing them against their reservation quantities. This is a dynamic game in which discounting plays an important role for its resolution; it is the effect of discounting that leads the agents into accepting an offer in the long run, so that the game terminates in finite time. The ability of the agents to wait was related to the bargaining power of the agents, the agent that can wait longer (has lower discount factor) has a stronger bargaining position than the other. Furthermore, in the appropriate limit, the Rubinstein model leads to the Nash bargaining solution, a remarkable fact that may lead to considering the Rubinstein model as a strategic justification for the Nash bargaining model. This is the limit as the absolute magnitudes of the frictions in the bargaining process are small. In turn, the Rubinstein model has ignited a large strand of important literature on bargaining, dealing with the evolution of the bargaining process both in theory and in applications (see e.g. \cite{Muthoo}).

In this paper we propose an alternative model for bargaining, based on the theory of projected gradient dynamical systems, for the determination of the price of a contingent asset whose real value depends on the future state of the economy. We assume that the two agents have their personal valuations of the asset in terms of reservation prices determined by subjective expected utility functions calculated under different probability measures concerning the economy. The agents are not strict concerning the probability of the future states of the world and are willing to update their beliefs on this matter, in an alternating fashion so that agreement is reached within a chosen time horizon. This is not an unreasonable assumption, as the agents may be assumed to derive some sort of information about their opponents beliefs from the offers they make concerning the reservation price of the asset. The update of beliefs is made in such a way that the agents deviate as small as possible from their initial beliefs, at each play of the game, while the update is made so that the new reservation prices are now closer to reaching agreement. This setup models the above quote from Muthoo according to which the two agents have common interest to cooperate.

The strategy space for the game is the space of probability measures on the future state of the economy. On account of that, the mechanism of the game is expressed as a projected gradient dynamical system in the space of discrete measures. This is related to a minimization problem -- that of minimizing a function modeling the joint deviation of the two agents new probability measures (beliefs) from their initial probability measures concerning the future state of the economy. Using techniques from convex dynamics and Lyapunov functions, we show that this projected gradient dynamical system converges to a fixed point in the space of measures which in turn provides a unique common price for the asset. 

This paper is organized as follows. In section \ref{setting} we describe the setup we will work on throughout the paper and formulate the bargaining scheme as a minimization problem, identifying its dual problem and possible alternative formulations as a penalty methods. In section \ref{Opt_prob} we discuss the convexity properties of the reservation prices of the buyer and the seller of the contingent claim and provide conditions for the minimization problem introduced in the previous section to have unique solutions. Section \ref{DUS} is devoted to the analysis of a class of discrete-time projected dynamical systems modeling the bargaining scheme as a sequence of ``small steps'', each given by an iteration of a projected dynamical system which is related to a gradient scheme associated with a penalized optimization problem introduced in section \ref{setting}. Namely, we study the robustness and stability of a general class of discrete-time projected dynamical systems, and then apply these general results to the family of dynamical systems modeling the agents bargaining process. In section \ref{exp_ut}, we restrict ourselves to the interesting case where the two agents involved in the trade make their decisions using exponential utility functions. Such utility functions satisfy the usual Inada conditions under which the results of the preceding sections were proved, providing a clear example and strong motivation for such results.

\section{Problem formulation} \label{setting}

Consider two agents, agent A being the seller of a (contingent) claim  and agent B the buyer. The payoff of the claim depends on the future state of the economy, $\omega_{k}$, $k=1,\cdots, K$, and is considered as a random variable $F$ defined on $(\Omega,\Fcal)$ where $\Omega$ is the set of states of the world and $\Fcal$ a $\sigma$-algebra on $\Omega$. We allow the two agents to have subjective beliefs concerning the future states of the world which are in principle different, and are modeled by two probability measures $Q_A$ and $Q_B$, on $(\Omega, \Fcal)$ respectively.  The beliefs of the agents $Q_{\beta}$, $\beta=A,B$, determine their views concerning the value of the claim $F$, therefore determine the bid and the ask price respectively for $F$.  We assume that the agents are not firm regarding their beliefs $Q_{\beta}$ and are willing to update them as part of a bargaining scheme whereby they exchange beliefs and reach an agreement on the price of the asset.

This situation is common in various situations in economics. For instance one may consider the case where $F$ is a financial asset e.g. a derivative asset in an incomplete market. Then, the use of the underlying financial market does not dictate the existence of a unique measure which is defined by general equilibrium arguments (Arrow-Debreu measure) and which acts at the common belief for both agents (as would happen in the case of complete markets) and subjective beliefs of the agents as well as personal tastes with respect to risk enter the decision process. Our approach is even more relevant in other cases; e.g. in cases where an asset is traded outside the market or in bilateral agreements such as mergers and acquisitions where subjectivity plays a very important role in the determination of the price.

We assume that the agents use utility pricing considerations to specify the price of the asset to be traded. In particular, both agents report an expected utility  $E_{Q_{\beta}}[U_{\beta}]$ under their subjective probability measures $Q_{\beta}$, $\beta=A,B$. The utility functions $U_{\beta}: \R \rightarrow \R$ satisfy the usual 
 Inada conditions \cite{Inada}, that is, the utility functions have value zero when $x=0$, are strictly increasing, strictly concave and continuously differentiable, and their first derivatives satisfy the following asymptotic conditions
\begin{equation*}
\lim_{x\rightarrow -\infty} U_{\beta}'(x) = +\infty \ ,  \qquad \lim_{x\rightarrow +\infty} U_{\beta}'(x) = 0 \ , \qquad \beta=A,B \ .
\end{equation*}

The reservation price $P_{\beta}$ for agent $\beta$ is given by the solution of the following equation
\begin{equation}\label{res_pr_AB}
U_{\beta}(w_{\beta}) -r_{\beta} = E_{Q_{\beta}}	\left [U_{\beta}(w_{\beta} + \ell_{\beta} P_{\beta} - \ell_{\beta} F)\right ] \ , 
\end{equation} 
where $\ell_{A}=1$ and $\ell_{B}=-1$. In the above equation
 $w_{\beta} \in \R$ is the initial wealth of agent ${\beta}$, considered to be fixed and deterministic, and $r_{\beta}\in R_{\beta} \subseteq \R_0^+$ can be considered to be the level of risk undertaken by agent ${\beta}$. If $r_{\beta}$ is zero then we are back to the standard reservation price whereas if $r_{\beta}$ is positive, we model the situation where agent ${\beta}$ is willing to go below her initial utility level so that the transaction is made possible. The introduction of $r_{\beta}$ can also model the fact that the utility function may not be known exactly and can even be a random utility. 
The asymmetry between $A$ and $B$ introduced by $\ell_{\beta}$ is  due to the opposite position undertaken by the two agents on the asset to be traded. Therefore an alternative interpretation of $r_{\beta}$ can be as a measure of how much agent $A$ ($B$) is willing to lower (raise) her reservation price so that they both agree upon an acceptable price for $F$.  We wish however to stress the fact that the proposed bargaining dynamics do {\bf not} depend on the introduction of the term $r_{\beta}$, which can be taken without loss of generality to be equal to zero \footnote{ This will correspond to the case where no misspecification of utility is allowed in the model.}. 
In general, the reservation prices $P_{\beta}$ will depend on $Q_{\beta}$ and $r_{\beta}$, i.e., $P_{\beta}:=P_{\beta}(r_{\beta},Q_{\beta})$.

We now propose a bargaining scheme between the two agents.
 We suppose that throughout the bargaining period each agent $\beta$ keeps the level of risk $r_{\beta}$ fixed. Assume that the transaction takes place on a fixed interval $[0,T]$ and the agents report reservation prices $P_{\beta}(t)$ at discrete times $t\in {\mathbb N}$. At time $t=0$ the two agents state their reservation prices $P_{\beta}(0)$  according to their initial beliefs $Q_{\beta}(0)$, $\beta=A,B$. Assume that these initial beliefs are such that $P_A(0)> P_B(0)$, therefore the two agents do not agree on a transaction for the exchange of the asset $F$.  As the beliefs of the agents concerning future states of the world are not firm, they are both willing to review their subjective probabilities (beliefs) $Q_{\beta}$ so as to reach a new reservation price, hoping to reach an agreement.  Then, for each instant of time $t\in {\mathbb N}$, agent $\beta$ ($\beta=A,B$) adopts the set of beliefs $Q_{\beta}(t)$ and states her reservation price $P_{\beta}(t)$ as the solution of equation \eqref{res_pr_AB} with $Q_{\beta}$ replaced by $Q_{\beta}(t)$.   A transaction takes place at time $t$ if and only if $P_A(t)\le P_B(t)$. By the monotonicity properties of the utility functions $U_A$ and $U_B$ one can easily see that the above constraint is binding and therefore we will only consider the case $P_A(t) = P_B(t)$.

The update of subjective probabilities of the agents each period is done in such a manner that given $r_{\beta}$ (possibly $0$) the reservation prices should eventually converge, therefore 
 the two agents are expected to update their beliefs until reaching some set of beliefs $Q_A^*$ and $Q_B^*$ for which the corresponding reservation prices are such that the equality $P_A(r_A,Q_A^*)=P_B(r_B,Q_B^*)$ is satisfied and trade takes place. The agents are willing to deviate as little as possible from their original beliefs \cite{BPPXY_2010} therefore it is natural to assume that $(Q_A^*,Q_B^*)$ is as close as possible to $(Q_{A}(0),Q_{B}(0))$. 
 This remark indicates that the choice of $(Q_A^*,Q_B^*)$ can be obtained as the solution of a properly chosen optimization problem. This optimization problem will also lead us to a dynamical theory of bargaining. 
 
  Let $\Delta^K$ be the unit simplex in $\R^K$ and let $\Dcal=\Delta^K\times\Delta^K$ be the state space for the agents beliefs. The distance of a subjective probability measure $Q_{\beta}$ from the reference measure $Q_{\beta}(0)$ is modeled by a strictly convex function 
  $\psi_{\beta} : \Delta^{K} \times \Delta^{K} \rightarrow \R^{+}$, such that $\psi(x,x)=0$. The value of the function $\psi_{\beta}(Q_{\beta},Q_{\beta}(0))$  models how agent $\beta$ evaluates the distance of the new subjective probability measure $Q_{\beta}$ from the original belief $Q_{\beta}(0)$. The smaller this value the more willing is agent $\beta$ to adopt the new subjective probability measure $Q_{\beta}$ in place of her original choice $Q_{\beta}(0)$. Of course if it were not for the need to change $Q_{\beta}$ to get closer to an agreement in price, the agent would stick to her original choice. The agents do not decide on their own, they have to take into account the beliefs of each other. Therefore, in reaching a common agreement  the unwillingness of both to deviate too much from their original subjective probabilities should be taken into account. This may be modeled by the convex combination
  \begin{equation}
  \lambda  \psi_A(Q_A,Q_A(0)) + (1-\lambda) \psi_B(Q_B,Q_B(0))
  \nonumber
  \end{equation}
  where $\lambda \in [0,1]$ is a real parameter that models which agent is allowed to deviate less from her original subjective probability. If $\lambda=0$ then agent $B$ has the dominant role in this transaction as it is only her unwillingness to deviate from original beliefs that is taken into account, whereas no respect for agent A feelings towards this matter is paid. If $\lambda=1$ the opposite happens. Any other choice of $\lambda \in (0,1)$ corresponds to ``sharing'' the discontent caused by the change of subjective probability measure between the two agents. Therefore, in some respect $\lambda$ may be thought of as a relative measure of ``bargaining power'' of two agents.
  
 Therefore, according to the above discussion, for some fixed levels of risk $r_A$ and $r_B$, we define the optimal beliefs $(Q_A^*,Q_B^*)\in\Dcal$ as the solution to the following minimization problem   
\begin{eqnarray}\label{PRIMAL} 
&&\min_{(Q_A,Q_B)\in\Dcal} \; \lambda \psi_A(Q_A,Q_A(0)) + (1-\lambda) \psi_B(Q_B,Q_B(0)) \nonumber \\
&& \mbox{subject to}\\
&&P_A(r_A,Q_A) \le P_B(r_B,Q_B). \nonumber 
\end{eqnarray}

Resorting to an analogy between the optimization problem \eqref{PRIMAL} and the classical utility maximization problem \cite[Ch. 3]{MasCollel_Whinston_Whinston}, it is possible to check that the dual problem of \eqref{PRIMAL} can be written as
\begin{eqnarray}\label{DUAL} 
&&\max_{ (Q_A,Q_B)\in\Dcal  } P_B(r_B,Q_B) - P_A(r_A,Q_A)  \nonumber \\
&&\mbox{subject to }\\
&& \lambda \psi_A(Q_A,Q_A(0)) + (1-\lambda) \psi_B(Q_B,Q_B(0)) \le \beta . \nonumber  
\end{eqnarray}
Note that the dual problem \eqref{DUAL} corresponds to maximizing the difference of the reservation prices by the buyer and by the seller, thus ensuring the trade of the contingent claim, under the constraint that the agents beliefs do not change more than a fixed value as measured by the functions $\psi_A$ and $\psi_B$ \cite{BPPXY_2010}. As is well known, duality implies that the solutions of the two problems above are intimately related. See Bertsekas textbook \cite{Bertsekas} for further details on nonlinear optimization and duality.

In section \ref{DUS} we will address the minimization problem \eqref{PRIMAL} using a projected gradient scheme associated with a couple of penalized optimization problems. The first such problem is of the form
\begin{eqnarray*}\label{PRIMAL_penalty} 
\min_{(Q_A,Q_B)\in\Dcal}\; \lambda \psi_A(Q_A,Q_A^0) + (1-\lambda) \psi_B(Q_B,Q_B^0) +\frac{1}{\epsilon}\Phi(P_B(r_B,Q_B) - P_A(r_A,Q_A)) \ , 
\end{eqnarray*}
where $\Phi:\R\rightarrow\R$ is a strictly convex function with a (global) minimum $\Phi(0)$ and $\epsilon$ is a small parameter.  The second penalized optimization problem is of the form 
\begin{eqnarray*}
\min_{(Q_A,Q_B)\in\Dcal}\; && \lambda \psi_A(Q_A,Q_A^0) + (1-\lambda) \psi_B(Q_B,Q_B^0) \\
&& + \frac{1}{\epsilon}\Phi\left[P_B(r_B,Q_B) - \lambda P_A(r_A,Q_A^0)-(1-\lambda)P_B(r_B,Q_B^0)\right]  \\
&& + \frac{1}{\epsilon}\Phi\left[\lambda P_A(r_A,Q_A^0)+(1-\lambda)P_B(r_B,Q_B^0)-P_A(r_A,Q_A)\right]  \nonumber \ ,
\end{eqnarray*}
where $\Phi$ and $\epsilon$ are as described above.

By taking the gradient of the functionals to be minimized and projecting it to the space of discrete probability measures in a appropriate way, the two penalized optimization problems above enable us to define discrete-time dynamical systems which, under appropriate conditions, model the bargaining process under which the two agents agree on a unique price to trade the contingent claim $F$. We note that such projected gradient dynamical systems can be idealized as a repetition in time of the above optimization problems, with the agents $\beta$ choosing the subjective probability $Q_{\beta}(t+1)$ so as to get closer to an agreement in reservation prices but at the same time with the minimum possible common deviation from their initial beliefs $Q_{\beta}(0)$. 

\begin{remark}
We remark that $1/\epsilon$ can be thought of as a discount factor in the model, that is, a factor modeling the degree of impatience for the agents to reach an agreement. The smaller $\epsilon$ is, the more impatient the agents are, and this in some sense characterizes the time for an agreement to be reached. Note that discounting plays a very important role in bargaining models in general, e.g. \cite{Muthoo} and references therein. 
\end{remark}

Even though the two penalized optimization problems above share the same form, i.e. a sum of a convex linear combination of ``distance functions'' with a penalty term, the differences in their penalty terms (dependence on $\lambda$ and past beliefs) provide very distinct dynamical behaviors for the corresponding projected gradient systems. We will provide a detailed description of such dynamics in section \ref{DUS}.

\section{Preliminary results on  reservation prices and  the optimization problem}\label{Opt_prob}

We will devote the beginning of this section to an analysis of several properties of the agents reservation price functions. We will then use these properties to provide conditions for the existence and uniqueness of solutions to the minimization problem \eqref{PRIMAL} where one looks for optimal beliefs for fixed levels of risks.

\subsection{Reservation price functions properties} 
In this section we provide a detailed analysis of the agents reservation price functions properties. We prove that such functions are well-defined before moving on to list some of its convexity properties. We end this section with results concerning the regularity of these functions.

The standing assumption throughout this section is that the utility functions of the agents $U_{\beta}$ satisfy the Inada conditions. To ease notation we define
\begin{equation}\label{functional_definition}
L(Q_A,Q_B\,;\, Q_A^0,Q_B^0,\lambda):=  \lambda  \psi_A(Q_A,Q_A^0) + (1-\lambda) \psi_B(Q_B,Q_B^0)
\end{equation}
which is considered as a function of $(Q_A,A_B) \in \Dcal$, the terms after the semicolon considered as parameters.  
We will also use the notation $\ell_{A}=1$ and $\ell_{B}=-1$.

\begin{lemma}\label{P_wd}
The following hold for the agents reservation prices:
\begin{enumerate}
\item The functions $P_{\beta} :\R_0^+\times\Delta^K\rightarrow\R$ are well defined.
\item The functions $\ell_{\beta} P_{\beta}$ are convex functions of $r_{\beta} \in\R_0^+$ for fixed beliefs $Q_{\beta} \in\Delta^K$.
\item The functions $\ell_{\beta} P_{\beta}$ are concave functions of  $Q_{\beta} \in\Delta^K$ for fixed values of $r_{\beta}\in\R_0^+$;
\end{enumerate}
\end{lemma}

The proof of this lemma is technical and is included in the Appendix.

A  consequence of lemma \ref{P_wd} and  is the key fact that when taken as functions of both the agents risk levels and the agents beliefs, the reservation prices $P_A$ and $P_B$ are neither concave, nor convex, limiting the array of techniques available to address an eventual problem where both a convex linear combinations of the agents levels of risk and the functional in \eqref{PRIMAL} were to be jointly minimized. 

We now provide results concerning boundedness and differentiability of the reservation price functions. We start by proving that the agents reservation prices are differentiable with respect to their beliefs and risk levels, computing also the corresponding partial derivatives.

\begin{lemma}\label{P_dif}
 There exists an open neighborhood of $\R_0^+\times\Delta^K$ where $P_{\beta}$  are continuously differentiable with partial derivatives given by
\begin{eqnarray*}
\frac{\partial P_{\beta}}{\partial Q_{\beta}^k}(r_{\beta},Q_{\beta}) & =& -\ell_{\beta} \frac{U_{\beta}(w_{\beta} +\ell_{\beta} P_A-\ell_{\beta} F[k])}{E_{Q_{\beta}}[U_{\beta}'(w_{\beta}+\ell_{\beta} P_A-\ell_{\beta} F)]} \ , \quad k=1,\ldots, K  \nonumber \\
\frac{\partial P_{\beta}}{\partial r_{\beta}}(r_{\beta},Q_{\beta}) & =& -\ell_{\beta} \frac{1}{E_{Q_{\beta}}[U_{\beta}'(w_{\beta}+\ell_{\beta} P_{\beta}-\ell_{\beta} F)]} 
\end{eqnarray*}
\end{lemma}
\begin{proof}
We prove the statement for the price function of agent A, the proof for agent B being similar.
Recall from the proof of lemma \ref{P_wd} that the expected utility function of agent A for the contingent claim $F$, $\overline{U}_A$ is a strictly increasing concave function of $x$. Since $\overline{U}_A$ is a convex linear combination of differentiable functions, it is also differentiable. We obtain that the expected utility $\overline{U}_A$ is invertible with differentiable inverse. The result then follows by the implicit function theorem, the expressions for the partial derivatives of $P_A$ with respect to the components of $Q_A$ being obtained through implicit differentiation of \eqref{P_der_1} and similarly for the partial derivative of $P_A$ with respect to $r_A$.
\end{proof}

\begin{remark}
We note that:
\begin{itemize}
\item[(i)]  The knowledge of the partial derivatives of the price functions restricted to $\R^+\times\Delta^K$ provides us with information regarding the sensitivity of the price functions with respect to changes in the agents beliefs and levels of risk which is  a key ingredient for the setup of the projected dynamical systems of section \ref{DUS}.
\item[(ii)] Since the utility functions derivatives  $U'_A$ and $U'_B$ are strictly positive, we obtain that the reservation price of the seller $P_A$ is a decreasing function of the amount of risk $r_A$ undertaken by the seller of the contingent claim and, similarly, the reservation price of the buyer $P_B$ is an increasing function  of the amount of risk $r_B$ undertaken by the buyer. 
\item[(iii)] The equalities 
\begin{eqnarray*}
\frac{\partial P_A}{\partial Q_A^k}(r_A,Q_A) &=& U_A(w_A+P_A-F[k])\frac{\partial P_A}{\partial r_A}(r_A,Q_A) \ , \quad k=1,\ldots, K  \nonumber \\
\frac{\partial P_B}{\partial Q_B^k}(r_B,Q_B) &=& U_B(w_B+P_B-F[k])\frac{\partial P_B}{\partial r_B}(r_B,Q_B) \ , \quad k=1,\ldots, K  
\end{eqnarray*}
provide us with an alternative interpretation for the update of beliefs under fixed levels of risk: the sensitivity of the agents price functions to such changes is proportional to the sensitiveness of the agents price to changes in their risk levels under a fixed set of beliefs, the proportionality constant being given by the utility level in the corresponding state of the world. 
\end{itemize}
\end{remark}

The next result provides upper and lower bounds for the reservation prices of the seller and the buyer of the contingent claim as well as a criteria for the existence of beliefs $Q_A$ and $Q_B$ and levels of risk $r_A$ and $r_B$ under which the two agents agree on the price for the contingent claim.
\begin{lemma}\label{lemma_bounds}
Let us introduce the notation $\overline{F} = \underset{\omega\in\Omega}{\max}\; F[\omega] $ and $\underline{F}=\underset{\omega\in\Omega}{\min}\; F[\omega]$.\\
The reservation prices satisfy the inequalities
\begin{eqnarray*}
U_A^{-1}\left(U_A(w_A)-r_A\right) - w_A + \underline{F} \; \le &P_A& \le \; U_A^{-1}\left(U_A(w_A)-r_A\right) - w_A + \overline{F} \nonumber \\
w_B - U_B^{-1}\left(U_B(w_B)-r_B\right) + \underline{F} \; \le &P_B& \le \; w_B - U_B^{-1}\left(U_B(w_B)-r_B\right) + \overline{F} \ .
\end{eqnarray*}
Moreover, there exist $Q_A,Q_B\in\mathrm{int}(\Delta^K)$ and $r_A,r_B\in\R_0^+$ such that $P_A(r_A,Q_A)=P_B(r_B,Q_B)$ if and only if the following condition holds
\begin{equation}\label{dif_ine}
\overline{F} - \underline{F} > w_B - U_B^{-1}\left(U_B(w_B)-r_B\right) + w_A - U_A^{-1}\left(U_A(w_A)-r_A\right) \  .
\end{equation}
\end{lemma}
The proof is given in Appendix \ref{proof_lemma_bounds}

\begin{remark} \hfill
\begin{itemize}
\item[(i)] The proof of the previous lemma also ensures that inequality \eqref{dif_ine} holds if there exist some $Q_A,Q_B\in\Delta^K$ such that $P_A(r_A,Q_A)>P_B(r_B,Q_B)$.
\item[(ii)] The right hand side of inequality \eqref{dif_ine} tends to zero as $r_A,r_B\rightarrow 0$ or $w_A,w_B\rightarrow +\infty$, making this a trivial condition in these cases. Moreover, we note that
\begin{eqnarray*}
\sup_{Q_A\in\Delta^K} P_A(r_A,Q_A) &\le& \sup_{Q_B\in\Delta^K} P_B(r_B,Q_B) \nonumber \\ 
\inf_{Q_A\in\Delta^K} P_A(r_A,Q_A) &\le& \inf_{Q_B\in\Delta^K} P_B(r_B,Q_B) \nonumber \ ,
\end{eqnarray*}
the inequalities above being a consequence of the non-negative levels of risk $r_A$ and $r_B$ undertaken by the agents. In particular, positive levels of risk lead to an higher upper bound for the buyer reservation price and a smaller lower bound for the seller reservation price.
\item[(iii)] As a consequence of the previous lemma we obtain that for the agents to agree on a unique price for the contingent, their wealths and levels of risk must satisfy some boundedness conditions. In particular if
 inequality \eqref{dif_ine} is satisfied, we have that the agents levels of risk are such that $r_{\beta} \in R_{\beta}$  where $R_{\beta}$, $\beta=A,B$ are bounded sets. The result follows easily from the observation that the right hand side of inequality \eqref{dif_ine} is an increasing function of $r_{\beta}$, therefore $R_{\beta} \subseteq\R_0^+$ must be also bounded above.
\end{itemize}
\end{remark}

A final result to conclude this section concerns the regularity of the price functions derivatives with respect to the beliefs of the agents under some mild additional assumptions on the agents utility functions, wealth and levels or risk. 

\begin{lemma}\label{DP_Lip}
Assume $U'_{\beta}$ are Lipschitz continuous functions and that $r_{\beta} \in R_{\beta}$ where $R_{\beta}$ is a bounded subset of $\R_0^+$. Then,   the partial derivatives of the price functions $P_{\beta}$ with respect to $(r_{\beta},Q_{\beta})$  are Lipschitz continuous on $\R_0^+\times\Delta^K$, for $\beta=A,B$.
\end{lemma}
The proof of the Lemma is given in Appendix \ref{proof_DP_Lip}.

\subsection{The minimization problem solvability and solution dependence on parameters}

In this section our focus lies on the minimization problem \eqref{PRIMAL} concerning the choice of optimal beliefs $(Q_A^*,Q_B^*)\in\Delta^K\times\Delta^K$ for fixed fixed positive levels of risk $r_A$ and $r_B$ undertaken by the two agents. The next theorem provides conditions for the existence and uniqueness of a solution for such minimization problem.

\begin{theorem}\label{opt_existence}
Assume that the agents initial beliefs $Q_A^0,Q_B^0\in\Delta^K$ are such that $P(r_A,Q_A^0)>P(r_B,Q_B^0)$ for fixed levels of risk $r_A,r_B\in\R_0^+$ and wealths $w_A,w_B\in\R$.

Then, for every weight $\lambda\in (0, 1)$ the minimization problem \eqref{PRIMAL} has a unique solution $(Q_A^*,Q_B^*)\in\Delta^K\times\Delta^K$ providing a unique price 
\begin{equation*}
P^*(r_A,r_B,\lambda) = P_A(r_A,Q_A^*)= P_B(r_B,Q_B^*) 
\end{equation*}
for the transaction of the contingent claim. 

Furthermore,   the solution $(Q_A^*,Q_B^*)$ depends continuously on the agents relative weight $\lambda$ and risk levels $r_A$ and $r_B$. 
\end{theorem}

The proof uses standard arguments is given in Appendix \ref{proof_opt_existence} for completeness.

The next result is a clear consequence of the previous theorem and duality theory for optimization problems.
\begin{corollary}   
Under the conditions of theorem \ref{opt_existence}, the (dual) maximization problem \eqref{DUAL} has a unique solution and at the optimal point $(Q_A^*,Q_B^*)\in\Delta^K\times\Delta^K$ we have that 
\begin{equation*}
P_B(r_B,Q_B^*) - P_A(r_A,Q_A^*) = 0 \ .
\end{equation*}
\end{corollary}

\begin{remark}
On account of the fact that the reservation prices $P_{\beta}$  are neither concave nor convex as joint functions of the agents risks and beliefs the treatment of 
the ``enlarged'' optimization problem
\begin{eqnarray*} 
\min_{(r_{\beta},Q_{\beta})\in(\R_0^+\times\Delta^K)^2} \; \lambda r_A +(1-\lambda)r_B ´+ \lambda \psi_A(Q_A,Q_A^0) + (1-\lambda) \psi_B(Q_B,Q_B^0)
\end{eqnarray*} 
subject to the constraint $P_A(r_A,Q_A) \le P_B(r_B,Q_B)$ 
is not a straightforward issue and may not lead to unique solutions.
\end{remark}

\section{A projected gradient scheme in discrete time} \label{DUS}

In this section we introduce a discrete-time dynamical system modeling the bargaining process under which the two agents update their beliefs in order to reach a common price for the contingent claim to be traded. This dynamical system is defined as a projected gradient scheme which approximates solutions of minimization problems of the same form as the minimization problem \eqref{PRIMAL}. Its phase space is $\Dcal:=\Delta^K\times\Delta^K$ and its iterates represent the evolution with time of the pair of beliefs $(Q_A(t),Q_B(t))$, $t\in\N$.

Suppose that at the initial time $t=0$ the eventual seller and buyer of the contingent claim have, respectively, initial beliefs $Q_A(0),Q_B(0)\in\Delta^K$ satisfying $P_A(r_A,Q_A(0))>P_B(r_B,Q_B(0))$. At each instant of time $t\in\N$, the two agents update their beliefs from the probability measures $(Q_A(t),Q_B(t))\in\Dcal$ to new probability measures $(Q_A(t+1),Q_B(t+1))\in\Dcal$ using a projected gradient scheme associated with the penalized unconstrained optimization problem for the function
\begin{equation}\label{LPDEF}
L_{P}(Q_{A},Q_{B}; Q_{A}^0, Q_{B}^0 , \lambda):= L(Q_{A},Q_{B}; Q_{A}^0, Q_{B}^0 , \lambda) + 
\frac{1}{\epsilon}\Phi(P_B(r_B,Q_B) - P_A(r_A,Q_A))
\end{equation}
where the penalty function $\Phi:\R\rightarrow\R$ is assumed to be a strictly convex continuously differentiable function with a (global) minimum $\Phi(0)$ and $\epsilon$ is a small parameter measuring the size of the penalty term. A simple example of a penalty function satisfying these conditions is provided by the quadratic function $\Phi(x)=x^2$, but many other choices are equally suitable for the role of penalty function, e.g. the logarithmic penalty function etc.

We will start by considering the unconstrained optimization problem 
\begin{equation}\label{PRIMAL_time}
 \min_{(Q_A,Q_B)\in\Dcal} \;\;  L_{P}(Q_{A},Q_{B}; Q_{A}^0, Q_{B}^0 , \lambda)
\end{equation}
associated with the penalized functional \eqref{LPDEF} and studying the relation between its solutions and the solutions of the minimization problem \eqref{PRIMAL}. Recall that theorem \ref{opt_existence} provides conditions under which the minimization problem \eqref{PRIMAL} above has a unique solution yielding a common price for the contingent claim. This also justifies the assumptions on the form of the penalty function: it is zero on the codimension one set of beliefs where the two agents reservation prices coincide, and it is positive where the two reservation prices are unequal, thus penalizing choices of beliefs which do not lead to a unique price for the contingent claim. We will prove below that under the conditions of theorem \ref{opt_existence} there exists a small enough $\epsilon_0>0$ such that the penalized problem \eqref{PRIMAL_time} has also a unique solution for every $\epsilon\in(0,\epsilon_0)$. Moreover, we also prove that the unique solution of \eqref{PRIMAL_time}  depends continuously on $\epsilon$ and tends to the solution of \eqref{PRIMAL} as $\epsilon\rightarrow 0$ for fixed initial beliefs $(Q_A^0,Q_B^0)\in\Dcal$.

The dynamical system modeling the bargaining process is then an iterative algorithm for solving the minimization problem \eqref{PRIMAL} by approximating the penalized  problem \eqref{PRIMAL_time} which then approximates \eqref{PRIMAL} for small values of $\epsilon$. Let us introduce the notation
\begin{equation*}
x(t)= (Q_A(t),Q_B(t)) \in \Dcal
\end{equation*}
and let $L_P:\Dcal^2\rightarrow\R$ denote the functional \eqref{LPDEF}, now expressed in terms of the variable $x$ as  $L_P(x(t+1),x(t))$, where the semicolon notation as well as the explicit dependence on $\lambda$ are dropped for simplicity.
The projected dynamical system on $\Dcal$ modeling the bargaining scheme for the agents to reach a unique price for the contingent claim is the implicitly defined discrete-time dynamical system given by
\begin{equation}\label{PGS}
x(t+1)= x(t) + \alpha \left(\Pi_{\Dcal}\left[x(t) - \nabla_{x(t+1)} L_P(x(t+1),x(t))\right] -x(t)\right)\ ,
\end{equation}
where $\nabla_{x(t+1)} L_P(x(t+1),x(t))\in\R^{2K}$ denotes the gradient of $L_P$ with respect to the variable $x(t+1)$ and $\PiD:\R^{2K}\rightarrow \Dcal$ is the projection operator onto $\Dcal$ defined by
\begin{equation*}
\PiD(x) = \underset{y\in\Dcal}{\textrm{argmin}}\left\|x-y\right\|  
\end{equation*}
where $\left\|\cdot\right\|$ denotes the euclidean norm in $\R^{2K}$. We believe it is worth to remark that since ``initial beliefs'' components are changed at each iteration, the dynamical system under consideration does not typically lead to a solution of the optimization problems \eqref{PRIMAL_time} or \eqref{PRIMAL} in a single iteration.

\begin{remark}
We note that for $x=(x_1,\ldots,x_K,x_{K+1},\ldots,x_{2K})\in\Dcal$, the $i$-th component of the projection of $x$ onto $\Dcal$, $\left(\PiD(x)\right)_i$, is given by
\begin{equation}\label{proj_op}
\left(\PiD(x)\right)_{i} =
\begin{cases}
x_i - \frac{1}{K}\sum_{j=1}^K x_j + \frac{1}{K}  &\ , \quad i=1,\ldots,K \\ 
x_i  - \frac{1}{K}\sum_{j=K+1}^{2K} x_j  + \frac{1}{K} &\ , \quad i=K+1,\ldots,2K  \ .
\end{cases} 
\end{equation}
\end{remark}

In section \ref{SSPOP} we provide conditions under which the minimization problem \eqref{PRIMAL_time} has a unique solution and characterize how such solution depends on the parameter $\epsilon$. We devote section \ref{SSPDS} to the study of general properties of a family of dynamical systems which includes \eqref{PGS} as a particular case. We apply these results to the special case of \eqref{PGS} in section \ref{SSBD1}. In section \ref{SSBD2} we consider a modification of \eqref{PGS} which has the key feature that the penalty function depends on the relative bargaining power of the two agents $\lambda$.

\subsection{Penalized optimization problem}\label{SSPOP}

Before moving on to the analysis of the qualitative properties of the dynamical system defined by \eqref{PGS}, we provide conditions ensuring the existence and uniqueness of solutions for the penalized minimization problem \eqref{PRIMAL_time}. Note that even though the functional to minimize in \eqref{PRIMAL_time} may not be convex, we are still able to obtain uniqueness for the solution of such problem, as  stated in the proposition below.

\begin{proposition}
Assume that the agents initial beliefs are such that $(Q_A^0,Q_B^0)\in\Dcal$ are such that $P(r_A,Q_A^0)>P(r_B,Q_B^0)$ for fixed levels of risk $r_A,r_B\in\R_0^+$ and wealths $w_A,w_B\in\R$.

There exists $\epsilon_0>0$ such that for every $\epsilon\in (0,\epsilon_0)$ the optimization problem \eqref{PRIMAL_time} has a unique solution which depends continuously on $\epsilon$. Moreover, this solution converges to the solution of the original constrained problem \eqref{PRIMAL} as $\epsilon\rightarrow 0$.
\end{proposition}
\begin{proof}
The existence of a solution to \eqref{PRIMAL_time} is clear. We deal only with the uniqueness of such solution in this proof. 

Since the agents initial beliefs $Q_A^0,Q_B^0\in\Delta^K$ are such that $P(r_A,Q_A^0)>P(r_B,Q_B^0)$, then inequality \eqref{dif_ine} holds and lemma \ref{lemma_bounds} guarantees the existence of $Q_A,Q_B\in\Delta^K$ such that $P_B(r_B,Q_B) = P_A(r_A,Q_A)$. A straightforward consequence is that $\Phi(P_B(r_B,Q_B) - P_A(r_A,Q_A))$ is not convex. In fact, $\Phi(P_B(r_B,Q_B) - P_A(r_A,Q_A))$ is convex on the subset of $\Dcal$ defined by
\begin{equation*}
S= \{(Q_A,Q_B)\in\Dcal: P_B(r_B,Q_B) \ge P_A(r_A,Q_A)\} 
\end{equation*}
but this is no longer true on the complement of $S$ on $\Dcal$. Before proceeding with the proof, we introduce some notation. Let $\lambda\in(0,1)$, $Q_A^*\in\Delta^K$ and $Q_B^*\in\Delta^K$ be fixed but arbitrary, and let $\overline{\Phi}:\Dcal\rightarrow\R$ be given by
\begin{equation*}
\overline{\Phi}(Q_A,Q_B) = \Phi(P_B(r_B,Q_B) - P_A(r_A,Q_A)) \ ,
\end{equation*}
define $\overline{\Psi}:\Dcal\rightarrow\R$ as 
\begin{equation*}
\overline{\Psi}(Q_A,Q_B) = \lambda \psi_A(Q_A,Q_A^*) + (1-\lambda) \psi_B(Q_B,Q_B^*) \ ,
\end{equation*}
and let $\overline{L}_P:\Dcal\rightarrow\R$ be given by
\begin{equation}\label{LPproof}
\overline{L}_P(Q_A,Q_B) = \overline{\Psi}(Q_A,Q_B) + \frac{1}{\epsilon}\overline{\Phi}(Q_A,Q_B)  \ .
\end{equation}

Since $\overline{\Psi}$ is a convex function in $\Dcal$, $\overline{L}_P$ is of the form \eqref{LPproof} and has at least one local minimum, then there exists $\delta>0$ such that ${\overline{L}_P}^{-1}(0,\delta)$ is a non-empty convex subset of $\Dcal$. 
Let us define $\epsilon_0$ as
\begin{eqnarray*}
\epsilon_0 = \sup\left\{\epsilon\in\R_0^+: \nabla_{(Q_A,Q_B)} \overline{L}_P \neq 0 \textrm{ for every } (Q_A,Q_B)\in (\Dcal)\backslash{\overline{L}_P}^{-1}(0,\delta) \right\}	\ .
\end{eqnarray*}
Note that $\epsilon_0$ is strictly positive since $\overline{L}_P$ is continuously differentiable on $\Dcal$ and non-constant, and thus it must have at most a finite number of critical values on a compact set. Then, for every $\epsilon\in(0,\epsilon_0)$ the optimization problem \eqref{PRIMAL_time} has a unique solution.

Berge's maximum theorem \cite[Ch. VI, Sec. 3]{Berge} guarantees continuity of the optimal solution with respect to $\epsilon$. Using proposition 4.2.1 of \cite{Bertsekas} we conclude that the solution of \eqref{PRIMAL_time} converges to the solution of the original problem \eqref{PRIMAL}.
\end{proof}
\begin{remark}
If we assume the functional to minimize in \eqref{PRIMAL_time} to be at least twice differentiable with respect to the updated beliefs $(Q_A,Q_B)$, we could use the implicit function theorem to obtain a generic condition under which the unique solution of \eqref{PRIMAL_time} is differentiable with respect to $\epsilon$ on a small neighborhood of the origin of the form $(0,\epsilon_0)$, for some $\epsilon_0>0$.
\end{remark}

\subsection{Projected dynamical systems} \label{SSPDS}
We now consider a discrete-time projected dynamical system implicitly\footnote{The proposed implicit scheme can be thought of as a backward Euler approximation of a continuous time gradient flow. Similar results would hold for an explicit scheme which corresponds to a forward Euler approximation of a continuous gradient flow, of course with the usual drawbacks of the forward Euler scheme inherited by the explicit scheme. On the other hand, if we wish to assign a behavioral interpretation to the gradient scheme we may always revert to the explicit scheme by a power series approximation of the implicit scheme.} defined by
\begin{equation}\label{PDS}
x(t+1) = x(t) + \alpha \left( \PiA [x(t) - \beta G(x(t+1),x(t))] - x(t)\right)  \ ,
\end{equation}
where $\alpha$ and $\beta$ are such that $0<\alpha<1$ and $\beta>0$, $\Dconv$ is a convex compact subset of $\R^n$, $G:\Dconv\times \Dconv\rightarrow\R^n$ is a continuous function on $\Dconv\times \Dconv$ and $\PiA:\R^n\rightarrow \Dconv$ is the projection operator onto $\Dconv$ defined by
\begin{equation*}
\PiA(x) = \underset{y\in \Dconv}{\textrm{argmin}}\left\|x-y\right\|  \ ,
\end{equation*}
where $\left\|\cdot\right\|$ denotes the euclidean norm in $\R^n$. Furthermore, we assume that every point in $\Dconv$ has a unique forward orbit, i.e. there exists a map $F:\Dconv\rightarrow \Dconv$ such that
\begin{equation}\label{PDS2}
x(t+1) = F(x(t)) \ .
\end{equation}
In the case where $\PiA\left[y-\beta G(x,y)\right]$ is a differentiable function, the assumption above reduces to assuming that the inequality 
\begin{eqnarray*}
\det\left(\Id_{n\times n} -\alpha D_{x}\PiA\left[y-\beta G(x,y)\right]\right) \neq 0 
\end{eqnarray*}
holds for every $x,y\in \Dconv$, where $\Id_{n\times n}$ stands for the identity matrix in $\R^n$ and $D_x$ denotes differentiation with respect to $x$.

The next  auxiliary result concerns invariance and stability of the convex set $\Dconv$ under the dynamics defined by \eqref{PDS}.

\begin{lemma}\label{PDS_inv} The following hold: \\
(i)  $\Dconv$ is invariant under the projected dynamical system \eqref{PDS}.\\
 (ii) If $x_0\notin \Dconv$, the orbit through $x_0$ of the projected dynamical system \eqref{PDS} converges exponentially fast to  $\Dconv$ as $t\rightarrow\infty$.
\end{lemma}
\begin{proof}
(i) Let $x(t)$ be an element of $\Dconv$ and rewrite \eqref{PDS} as
\begin{equation*}
x(t+1) = (1-\alpha)x(t) + \alpha  \PiA [x(t)-\beta G(x(t+1),x(t))]  \ .
\end{equation*}
Since $\PiA$ is a projection operator onto $\Dconv$, we get that $x(t+1)$ is a convex linear combination of two elements of $\Dconv$. Therefore, we conclude that $x(t+1)\in \Dconv$, proving invariance of $\Dconv$ under the dynamical system \eqref{PDS}.

(ii) The continuous function $V:\R^n\rightarrow\R$ given by
\begin{equation*}
V(x) =\left\|x-\PiA x\right\|^2
\end{equation*}
is a Lyapunov function \cite{Agarwal} for the discrete time dynamical system \eqref{PDS}. Note that $V(x)=0$ for every $x\in \Dconv$ and that $V(x)>0$ for every $x\notin \Dconv$. It remains to check that $V$ is strictly decreasing on orbits of \eqref{PDS} through points not in $\Dconv$. 

Take an arbitrary point $x_0\notin \Dconv$ and let $\{x_n:n\in\N\}$ be the sequence of points in the forward orbit of $x_0$ by the dynamics of \eqref{PDS}. We will prove that the quantity
\begin{equation}\label{Lyap1}
\Delta V(n):= V(x_{n+1})-V(x_n) = \left\|x_{n+1}-\PiA x_{n+1}\right\|^2 - \left\|x_{n}-\PiA x_{n}\right\|^2 
\end{equation}
is negative for every $n\in\N$.  We first note that
\begin{equation}\label{Lyap2}
x_{n+1} = (1-\alpha)x_{n} +\alpha \PiA [x_n-\beta G(x_{n+1},x_n)] \ .
\end{equation}
Since $\PiA$ is a projection operator, we get
\begin{eqnarray}\label{Lyap3}
\PiA x_{n+1} &=& \PiA \left[(1-\alpha)x_{n} +\alpha \PiA [x_n-\beta G(x_{n+1},x_n)] \right] \nonumber \\
&=& (1-\alpha) \PiA x_{n}  + \alpha \PiA [x_n-\beta G(x_{n+1},x_n)]   \ .
\end{eqnarray}
Putting together \eqref{Lyap2} and \eqref{Lyap3}, we conclude that 
\begin{equation*}\label{Lyap4}
x_{n+1}-\PiA x_{n+1} = (1-\alpha)\left(x_{n} -  \PiA x_{n}\right) \ .
\end{equation*}
We now use the previous equality to obtain
\begin{equation}\label{Lyap5}
\left\|x_{n+1}-\PiA x_{n+1}\right\|^2 - \left\|x_{n}-\PiA x_{n}\right\|^2  = ((1-\alpha)^2-1)\left\|x_{n}-\PiA x_{n}\right\|^2 \ .
\end{equation}
Substituting \eqref{Lyap5} in \eqref{Lyap1} and recalling that $0<\alpha<1$, we get 
\begin{equation*}
V(x_{n+1})-V(x_n) = ((1-\alpha)^2-1)\left\|x_{n}-\PiA x_{n}\right\|^2 < 0 \ ,
\end{equation*}
thus proving that orbits of \eqref{PDS} through points not in $\Dconv$ converge to $\Dconv$ exponentially fast.
\end{proof}

The next theorem ensures robustness of the dynamical system \eqref{PDS} under sufficiently small perturbations.
\begin{theorem}
Assume that the map $\overline{G}:\R^n\times \R^n\rightarrow \Dconv$ given by $\overline{G}(x,y) = \PiA\left[y-\beta G(x,y)\right]$ is differentiable. Then $\Dconv$ is a normally hyperbolic invariant manifold for the projected dynamical system \eqref{PDS}.
\end{theorem}
\begin{proof}
Choose a local trivialization $(x,y)\in\cup_j M_j \times E$, where $\cup_j M_j$ is an open cover of the convex compact set $\Dconv$ and $E$ is a Banach space measuring the displacement from $\Dconv$. Note that in such coordinates, the set invariant set $D$ is defined by the condition $y=0$. Moreover, in the local coordinates $(x,y)$ the spectrum of the normal component must have all eigenvalues inside the unity circle by lemma \ref{PDS_inv}(ii). Using a partition of the unity argument one can prove that this property is global, thus proving normal hyperbolicity of the set $\Dconv$.
\end{proof}
Using the previous theorem, we obtain that the dynamics of \eqref{PDS} are robust under sufficiently small perturbations, i.e. any sufficiently close dynamical system still has a normally hyperbolic invariant manifold which is close to $\Dconv$. Furthermore, if we assume that the perturbation preserves the compact convex set $\Dconv$, we obtain that $\Dconv$ is still a normally hyperbolic manifold for the perturbed dynamics. This property ensures that in a neighborhood of the dynamical system \eqref{PDS}, the dynamics are qualitatively similar, providing robustness to the model under analysis. 

The result below provides a characterization for the fixed points of the projected dynamical system \eqref{PDS}. 

\begin{proposition}\label{fp_car}
  The following statements hold:
\begin{itemize}
\item[(i)] if $G(x^*,x^*)=0$, then $x^*\in \Dconv$ is a fixed point of \eqref{PDS};
\item[(ii)] if $x^*\in \mathrm{int}(\Dconv)$ is a fixed point of \eqref{PDS}, then $G(x^*,x^*)=0$. 
\end{itemize}
\end{proposition}
\begin{proof}
The proof of item (i) is trivial verification.

In what concerns the proof of item (ii), we recall that every fixed point $x^*\in \Dconv$ lies on the non-empty subset $S\subset \Dconv$ defined by
\begin{equation*}
S= \left\{x\in \Dconv:  x = \PiA [x - \beta G(x,x)] \right\} \ .
\end{equation*}
Since $\PiA$ is a projection operator (see \cite{Kinderlehrer_Stampacchia}) every $x^*\in S$ must satisfy the inequality
\begin{equation*}
\left\langle x^*,\eta- x^*\right\rangle  \ge \left\langle  x^* - \beta G(x^*,x^*),\eta-x^*\right\rangle 
\end{equation*}
for every $\eta\in \Dconv$. Using the inner product linearity, we obtain that the inequality
\begin{equation*}
\left\langle G(x^*,x^*),\eta- x^*\right\rangle  \ge 0  
\end{equation*}
must be satisfied for every $\eta\in \Dconv$. We conclude that $G(x^*,x^*)=0$ for every $x^*\in S\cap\mathrm{int}(\Dconv)$. 
\end{proof}

The next theorem concerns the existence and stability of fixed points of \eqref{PDS}.

\begin{theorem}\label{global_conv}
 Let $G_D:\Dconv\rightarrow \Dconv$,  $G_D(x):=G(F(x),x)$, (see \eqref{PDS2})  be Lipschitz continuous. 
 
 Then, the projected dynamical system \eqref{PDS} has at least one fixed point which is stable in the sense of Lyapunov  for every $\alpha$ such that $0< \alpha < (1+\beta^2L^2/2)^{-1}$, where $L$ is the Lipschitz constant of $G_D$. 
\end{theorem}
\begin{proof}
Using lemma \ref{PDS_inv} we obtain that $\Dconv$ is invariant under the dynamics of \eqref{PDS}. Moreover, since $\Dconv$ is convex and compact, Brouwer's fixed point theorem ensures the existence of at least one fixed point $x^*$ for the dynamical system \eqref{PDS}. Note that by definition of the dynamical system \eqref{PDS}, every fixed point $x^*\in \Dconv$ lies on the non-empty subset $S\subset \Dconv$ defined by
\begin{equation*}\label{glb_conv_1}
S= \left\{x\in \Dconv:  x = \PiA [x - \beta G(x,x)] \right\} \ .
\end{equation*}

Take an arbitrary point $x_0\in \Dconv$ and let $\{x_n:n\in\N\}$ be the sequence of points in the forward orbit of $x_0$ by the dynamics of \eqref{PDS}. We will now see that the sequence $\{x_n\}_{n\in\N}$ converges to some point $x^*\in S$. For that purpose, we will prove that $V:\Dconv\times \Dconv\rightarrow\R$ given by
\begin{equation*}
V(x,y) =\left\|y-\PiA\left[ y - \beta G(x,y)\right]\right\|^2
\end{equation*}
is a Lyapunov function \cite{Agarwal} for the discrete time dynamical system \eqref{PDS}. We start by noting that $V(x^*,x^*)=0$ for every $x^*\in S$. Moreover, we have that 
\begin{eqnarray*}
V(x_{n+1},x_n) &=&  \left\|x_n-\PiA\left[ x_n - \beta G(x_{n+1},x_n)\right]\right\|^2 \nonumber \\
&=& \frac{1}{\alpha^2}\left\|x_{n+1}-x_n\right\|^2 \ .
\end{eqnarray*}
Therefore, we obtain that $V(x_{n+1},x_n)>0$ for every $n\in\N$ and every trajectory of \eqref{PDS} not containing a fixed point $x^*\in S$. It remains to check that $V$ is strictly decreasing on orbits of \eqref{PDS}, i.e. we will prove that the variation
\begin{equation*}
\Delta V(n) = V(x_{n+1},x_n) - V(x_n,x_{n-1})
\end{equation*}
is negative for every $n\in\N$. Using the projected dynamical system \eqref{PDS} definition, we obtain that
\begin{equation*}
\PiA\left[ x_n - \beta G(x_{n+1},x_n)\right] = \frac{x_{n+1}-(1-\alpha)x_n}{\alpha} \ .
\end{equation*}
Using the previous identity, we get that
\begin{eqnarray*}
V(x_{n+1},x_n) &=& \left\|x_{n} -  \PiA\left[ x_n - \beta G(x_{n+1},x_n)\right] \right\|^2 \nonumber \\
& = & \frac{1}{\alpha^2} \left\|x_{n+1} -x_n \right\|^2 \ .
\end{eqnarray*}
In a similar fashion, we obtain that
\begin{equation*}
V(x_{n},x_{n-1}) = \frac{1}{\alpha^2} \left\|x_{n} -x_{n-1} \right\|^2 \ ,
\end{equation*}
and thus, we get that
\begin{equation}\label{glb_conv_2}
\Delta V(n) = \frac{1}{\alpha^2}\left(\left\|x_{n+1} -x_n \right\|^2 - \left\|x_{n} -x_{n-1} \right\|^2 \right) \ .
\end{equation}
Using \eqref{PDS} twice, we obtain
\begin{eqnarray}\label{glb_conv_3}
x_{n+1} - x_n &=& (1-\alpha)(x_{n}-x_{n-1})\nonumber \\
&& +\alpha\left(\PiA [x_n-\beta G(x_{n+1},x_n)] - \PiA [x_{n-1}-\beta G(x_{n},x_{n-1})]\right) \ .
\end{eqnarray}
Combining \eqref{glb_conv_2} and \eqref{glb_conv_3} with the triangle inequality, we obtain the inequality
\begin{eqnarray*}
\lefteqn{\left\|x_{n+1} - x_n \right\|^2 - \left\|x_{n} -x_{n-1} \right\|^2 \le  -(1-(1-\alpha)^2)\left\|x_{n} -x_{n-1} \right\|^2 } \nonumber \\
&&+\alpha^2\left\|\PiA [x_n-\beta G(x_{n+1},x_n)] - \PiA [x_{n-1}-\beta G(x_{n},x_{n-1})]\right\|^2 
\end{eqnarray*}
Therefore, we have that
\begin{eqnarray}\label{glb_conv_5}
\Delta V (n) &\le&  \left(1-\frac{2}{\alpha}\right)\left\|x_{n} -x_{n-1} \right\|^2 \nonumber \\
&&+\left\|\PiA [x_n-\beta G(x_{n+1},x_n)] - \PiA [x_{n-1}-\beta G(x_{n},x_{n-1})]\right\|^2 	\ .
\end{eqnarray}
Noting that $\PiA$ is a projection and using once more the triangle inequality, we get 
\begin{eqnarray*}
\lefteqn{\left\|\PiA [x_n-\beta G(x_{n+1},x_n)] - \PiA [x_{n-1}-\beta G(x_{n},x_{n-1})]\right\|^2 \le }  \nonumber \\
&&\left\|x_n - x_{n-1}\right\|^2 +\beta^2 \left\|G(x_{n+1},x_n)]  - G(x_{n},x_{n-1})\right\|^2  	\ .
\end{eqnarray*}
Using the representation \eqref{PDS2} for the projected dynamical system \eqref{PDS}, we can write $x_{n+1}=F(x_n)$ and $x_n=F(x_{n-1})$. Substituting into the previous inequality, we obtain
\begin{eqnarray*}
\lefteqn{\left\|\PiA [x_n-\beta G(x_{n+1},x_n)] - \PiA [x_{n-1}-\beta G(x_{n},x_{n-1})]\right\|^2 \le }  \nonumber \\
&&\left\|x_n - x_{n-1}\right\|^2 +\beta^2 \left\|G(F(x_n),x_n)]  - G(F(x_{n-1}),x_{n-1})\right\|^2  	\ .
\end{eqnarray*}
Since we assume $G(F(x),x)$ to be $L$-Lipschitz continuous, we obtain the following estimate
\begin{equation}\label{glb_conv_6}
\left\|\PiA [x_n-\beta G(x_{n+1},x_n)] - \PiA [x_{n-1}-\beta G(x_{n},x_{n-1})]\right\|^2 \le (1+\beta^2L^2)\left\|x_n - x_{n-1}\right\|^2 	\ .
\end{equation}
Combining \eqref{glb_conv_5} with \eqref{glb_conv_6}, we obtain
\begin{equation*}
\Delta V (n) \le \left(2 +\beta^2L^2-\frac{2}{\alpha}\right)\left\|x_{n} -x_{n-1} \right\|^2 
\end{equation*}
We conclude that for every $\alpha\in (0,(1+\beta^2L^2/2)^{-1})$ we have that $\Delta V (n)<0$ for every $n\in\N$. Since this estimate does not depend on the initial point $x_0\in \Dconv$ then, for every $\alpha$ in the interval above, every trajectory of \eqref{PDS} must converge to some fixed point $x^*\in S$.
\end{proof}

The next result is a natural consequence of the two previous results. 

\begin{corollary}
Assume that the map $\tilde{G}:\Dconv\rightarrow \R^n$,  $\tilde{G}(x): = G(x,x)$  is the gradient of a function $V$ with a unique minimum. Then, the projected dynamical system \eqref{PDS} converges to the unique minimizer of $V$ as $t\rightarrow+\infty$. 
\end{corollary}

We should note that the result above includes the cases of $V$ being convex or quasi-convex function since such functions have a unique minimum in every compact convex set (see \cite{Boyd_Vandenberghe,Mangasarian} for further details on generalizations of the concept of convexity and its consequences for optimization theory)

\subsection{Bargaining dynamics for a $\lambda$-independent penalty function}\label{SSBD1}
We now return to the analysis of the projected dynamical system \eqref{PGS} associated with the penalized optimization scheme \eqref{PRIMAL_time}. 

\begin{lemma}\label{uni_sol}
The discrete-time dynamical system defined implicitly by \eqref{PGS} is well defined, i.e. to every $x_0\in\Dcal$ corresponds exactly one forward trajectory $(x_n)_{n\in\N}$ such that for every $n\in\N$ the pair $(x_{n},x_{n-1})$ satisfies the recurrence \eqref{PGS}.
\end{lemma}
\begin{proof}
Consider the equation
\begin{equation}\label{PGS_2}
y= (1-\alpha)x +  \alpha\;\PiD\left[x - \nabla_1 L_P(y,x)\right] \ ,
\end{equation}
where $\nabla_1 L_P(y,x)$ denotes the gradient of $L_P(y,x)$ with respect to its first component $y\in\Dcal$.

Clearly, to prove the present lemma it is enough to show that for every $x\in\Dcal$, equation \eqref{PGS_2} has a unique solution $y\in\Dcal$. Suppose this is not the case, i.e. there exists $x\in\Dcal$ such that equality \eqref{PGS_2} has two distinct solutions $y^1,y^2\in\Dcal$. Using \eqref{PGS_2}, we obtain immediately that
\begin{equation*}\label{PGS_3}
y^1-y^2=   \alpha\left(\PiD\left[x - \nabla_1 L_P(y^1,x)\right]- \PiD\left[x - \nabla_1 L_P(y^2,x)\right] \right) \ .
\end{equation*}
Let us write $x=(x_A,x_B)$, where $x_A=(x_{A_1},\ldots,x_{A_K})\in\Delta^K$ and $x_B=(x_{B_1},\ldots,x_{B_K})\in\Delta^K$, and similarly for $y^1,y^2$. Combining the previous equality with the definition of the projection operator $\PiD:\R^{2K}\rightarrow\Dcal$ given in \eqref{proj_op}, we get that
\begin{equation}\label{PGS_4}
(K-1)D_{A_i} - \sum_{j\ne i}^K D_{A_j} = \frac{K}{\alpha}(y^2_{A_i}-y^1_{A_i}) \ , \qquad i=1,\ldots, K  \ ,
\end{equation}
where $D_{A_i}$ denotes the $A_i$-component of $\nabla_1 L_P(y^1,x)-\nabla_1 L_P(y^2,x)$. A similar set of conditions hold for the $B$-components of $\nabla_1 L_P(y^1,x)-\nabla_1 L_P(y^2,x)$.

We now note that \eqref{PGS_4} is a system of linear equations on $D_{A_i}$, $i\in\{1,\ldots,K\}$, with rank $K-1$ which admits a solution if and only if 
\begin{equation*}
\sum_{i=1}^K y^2_{A_i}-y^1_{A_i} = 0 \ .
\end{equation*}
Since $y^1_A$ and $y^1_A$ both belong to $\Delta^K$, the previous equality is automatically satisfied. Without loss of generality, we can solve \eqref{PGS_4} for $D_{A_1},\ldots,D_{A_{K-1}}$ with respect to $y^1$, $y^2$ and $D_{A_{K}}$ obtaining 
\begin{equation*}
D_{A_i} = \frac{2}{\alpha}(y^2_{A_i}-y^1_{A_i}) + \frac{1}{\alpha}\sum_{j\ne i}(y^2_{A_j}-y^1_{A_j}) + D_{A_K} \ , \qquad i\in\{1,\ldots,K-1\} \ .
\end{equation*}
Recalling that $y^1_A,y^2_A\in\Delta^K$, we can rewrite the previous set of equalities as
\begin{equation*}
D_{A_i} - D_{A_K} = \frac{1}{\alpha}(y^2_{A_i}-y^1_{A_i})  \ , \qquad i\in\{1,\ldots,K-1\} \ ,
\end{equation*}
which contradicts the fact that $\nabla_1 L_P(y^1,x)-\nabla_1 L_P(y^2,x)$ does not depend on $\alpha$. Thus, we must have that $y^1=y^2$, concluding the proof.
\end{proof}

\begin{lemma}\label{Lip_lemma}
Let $\psi_A$, $\psi_B$ and $\Phi$ be strictly convex functions with Lipschitz continuous derivatives. Assume that the conditions of lemma \ref{DP_Lip} hold. Then:
\begin{itemize}
\item[(i)] the functional $L_P$  (see \eqref{LPDEF}) has Lipschitz continuous partial derivatives;
\item[(ii)] there exists $\alpha_0>0$ such that for every $0<\alpha<\min\{1,\alpha_0\}$ the map $F:\Dcal\rightarrow\Dcal$ implicitly defined by
\begin{equation}\label{PDS_5}
F(x) = (1-\alpha)x + \alpha \PiD [x - \nabla_{1} L_P(F(x),x)]  
\end{equation}
is Lipschitz continuous.
\end{itemize}
\end{lemma}
\begin{proof}
Item (i) follows from the assumptions regarding Lipschitz continuity of the derivatives of $\psi_A$, $\psi_B$ and $\Phi$, together with lemma \ref{DP_Lip}.

We will now prove item (ii). By lemma \ref{uni_sol}, the map $F$ determined by \eqref{PDS_5} is well defined. Let $x^1,x^2\in\Dcal$ be arbitrary and let $\left\|\cdot\right\|$ denote the Euclidean norm. To avoid making the notation cumbersome, we do not distinguish between Euclidean norms in spaces of different dimensions, since this should be clear to the reader.

Using \eqref{PDS_5} we obtain the estimate
\begin{eqnarray*}
\lefteqn{\left\|F(x^2)-F(x^1)\right\| \le (1-\alpha)\left\|x^2-x^1\right\| } \nonumber \\
&&+ \alpha \left\| \PiD [x^2 - \nabla_{1} L_P(F(x^2),x^2)] - \PiD [x^1 - \nabla_{1} L_P(F(x^1),x^1)] \right\| \nonumber 
\end{eqnarray*}
Since $\PiD$ is a contraction, we get 
\begin{eqnarray*}
\lefteqn{\left\|F(x^2)-F(x^1)\right\| \le (1-\alpha)\left\|x^2-x^1\right\| } \nonumber \\
&&+ \alpha \left\| (x^2-x^1) - (\nabla_{1} L_P(F(x^2),x^2)  - \nabla_{1} L_P(F(x^1),x^1)) \right\| \nonumber 
\end{eqnarray*}
By the triangle inequality, we arrive at
\begin{equation*}
\left\|F(x^2)-F(x^1)\right\| \le \left\|x^2-x^1\right\| + \alpha \left\|\nabla_{1} L_P(F(x^2),x^2)  - \nabla_{1} L_P(F(x^1),x^1)\right\| \nonumber 
\end{equation*}
Using item (i) of the present lemma, we get that there exists $L>0$ such that 
\begin{equation*}
\left\|F(x^2)-F(x^1)\right\| \le (1+\alpha L)\left\|x^2-x^1\right\|  + \alpha L \left\|F(x^2)-F(x^1)\right\|  
\end{equation*}
Thus, for every $\alpha< \min\{1,L^{-1}\}$ we have that 
\begin{equation*}
\left\|F(x^2)-F(x^1)\right\| \le \frac{1+\alpha L}{1-\alpha L}\left\|x^2-x^1\right\| \ ,
\end{equation*}
concluding the proof.
\end{proof}

The next result follows from the stability theory for discrete-time projected dynamical systems developed above.
\begin{theorem}\label{thm_conv}
Assume Lipschitz continuity of the functions $U_{\beta}', \psi_{\beta}', \Phi'$, $\beta=A,B$.
Let the agents initial beliefs $Q_{\beta}(0)\in\Delta^K$ be such that $P(r_A,Q_A(0))>P(r_B,Q_B(0))$ for levels of risk $r_{\beta} \in R_{\beta}\subseteq\R_0^+$ and fixed wealths $w_{\beta}\in\R$, $\beta=A,B$. 

Then, for every $\epsilon>0$ there exists $\alpha_0:=\alpha_0(\epsilon)>0$ such that for every $0<\alpha<\min\{1,\alpha_0\}$ every forward orbit of the projected gradient scheme \eqref{PGS} converges to some fixed point  $(Q_A^*,Q_B^*)\in\mathrm{int}(\Dcal)$. Moreover, every fixed point $(Q_A^*,Q_B^*)\in\mathrm{int}(\Dcal)$ is such that 
\begin{equation*}
P_B(r_B,Q_B^*) - P_A(r_A,Q_A^*) = 0 
\end{equation*}
and $(Q_A^*,Q_B^*,Q_A^*,Q_B^*)\in(\Dcal)^2$ is a minimizer of the functional $L_P$ (see \eqref{LPDEF}).
\end{theorem}
\begin{proof}
Since the agents initial beliefs $Q_A(0),Q_B(0)\in\Delta^K$ are such that 
\begin{equation*}
P(r_A,Q_A(0))>P(r_B,Q_B(0)) \ ,
\end{equation*}
then equality \eqref{dif_ine} holds and by lemma \ref{lemma_bounds} (see remark (iii) following the lemma) we must have that the agents levels of risk are such that $r_A\in R_A$ and $r_B\in R_B$, where $R_A$ and $R_B$ are bounded above. Recalling that $U_A$ and $U_B$ satisfy the usual Inada conditions with the extra requirement that $U_A$ and $U_B$ have Lipschitz continuous derivatives, then lemma \ref{DP_Lip} guarantees that the reservation price functions $P_A$ and $P_B$ have Lipschitz continuous partial derivatives on a open neighborhood of $\R_0^+\times\Delta^K$. Finally, noticing that $\psi_A$, $\psi_B$ and $\Phi$ are also assumed to be strictly convex functions with Lipschitz continuous derivatives, we obtain that the first part of the statement follows from lemmas \ref{uni_sol} and \ref{Lip_lemma}, and theorem \ref{global_conv}.

For the second part of the statement, note that if $x^*=(Q_A^*,Q_B^*)\in\mathrm{int}(\Dcal)$ is a fixed point of the projected gradient scheme \eqref{PGS}, then $\nabla_{1} L_P(x^*,x^*)=0$ by proposition \ref{fp_car}. From the definition of the functional $L_P$ in \eqref{PGS}, we get that $(Q_A^*,Q_B^*)\in\mathrm{int}(\Dcal)$ must satisfy the equalities
\begin{eqnarray*}
\Phi'\left(P_B(r_B,Q_B^*) - P_A(r_A,Q_A^*)\right)\nabla_{Q_A} P_A (r_A,Q_A^*) &=& 0 \nonumber \\
\Phi'\left(P_B(r_B,Q_B^*) - P_A(r_A,Q_A^*)\right)\nabla_{Q_B} P_B (r_B,Q_B^*) &=& 0 \ ,
\end{eqnarray*}
where $\nabla_{Q_A} P_A$ and $\nabla_{Q_B} P_B$ denote, respectively, the gradient vectors of $P_A$ and $P_B$ with respect to $Q_A$ and $Q_B$. Since the contingent claim $F$ is assumed to be non-constant, we obtain from lemma \ref{P_dif} that such gradient vectors must be non-zero and the equalities above reduce to the single equality $P_B(r_B,Q_B^*) - P_A(r_A,Q_A^*)=0$. Since we assume that $Q_A(0),Q_B(0)\in\Delta^K$ are such that $P(r_A,Q_A(0))>P(r_B,Q_B(0))$ then inequality \eqref{dif_ine} must hold and thus, by lemma \ref{lemma_bounds}, the equality above has a solution $(Q_A^*,Q_B^*)\in\Dcal$. We conclude the proof by noting that under the assumptions on the functions $\Phi$, $\psi_A$ and $\psi_B$, the functional $L_P$ is bounded below by zero and that every point $(Q_A^*,Q_B^*,Q_A^*,Q_B^*)\in(\Dcal)^2$ such that $P_B(r_B,Q_B^*) - P_A(r_A,Q_A^*)=0$ realizes the (global) minimum of $L_P$. 
\end{proof}

Concerning the previous theorem, we remark that under the assumptions of theorem \ref{thm_conv}, the asymptotic beliefs $(Q_A^*,Q_B^*)\in \Dcal$ determining the price $P^*=P_B(r_B,Q_B^*) = P_A(r_A,Q_A^*)$ at which the contingent claim $F$ will eventually be traded depend very strongly on the the agents initial beliefs $(Q_A(0),Q_B(0))\in \Dcal$ and rather weakly on the remaining model parameters. This undesirable behavior is due to the form of the  penalty function.

\subsection{Bargaining dynamics for a $\lambda$-dependent penalty function}\label{SSBD2} We will now introduce a discrete-time projected gradient system which preserves all the desired properties of \eqref{PGS} described above, but has several interesting additional properties. Our motivation for the introduction of yet another projected gradient system is based on the observation that the fixed points $(Q_A^*,Q_B^*)\in\Dcal$ of \eqref{PGS} yielding a unique price for the contingent claim $F$ depend rather weakly on the relative bargaining power of the two agents $\lambda$ and remaining parameters, and that such dependence becomes weaker with smaller (positive) values of $\epsilon$. Clearly, such behavior is due to the fact that the penalty term does not depend explicitly on $\lambda$. The bargaining dynamics described below are based on a a penalty term depending on both $\lambda$ and the previous period prices.

Keeping the notation $x(t)= (Q_A(t),Q_B(t)) \in \Dcal$ we
define the functional $L_P^\lambda:(\Dcal)^2\rightarrow\R$  by
\begin{eqnarray}\label{L_pen_lambda}
\lefteqn{L_P^\lambda(x(t+1),x(t)) := \lambda \psi_A(Q_A(t+1),Q_A(t)) + (1-\lambda) \psi_B(Q_B(t+1),Q_B(t)) }\nonumber\\
&& + \frac{1}{\epsilon}\Phi\left[P_B(r_B,Q_B(t+1)) - \lambda P_A(r_A,Q_A(t))-(1-\lambda)P_B(r_B,Q_B(t))\right]  \\
&& + \frac{1}{\epsilon}\Phi\left[\lambda P_A(r_A,Q_A(t))+(1-\lambda)P_B(r_B,Q_B(t))-P_A(r_A,Q_A(t+1))\right]  \nonumber \ ,
\end{eqnarray}
where $\Phi:\R\rightarrow\R$ is a strictly convex continuously differentiable function with a (global) minimum $\Phi(0)$ and $\epsilon$ is a small parameter measuring the size of the penalty terms.

As in the case of \eqref{PGS},  we define a projected dynamical system on $\Dcal$ by 
\begin{equation}\label{PGS_lambda}
x(t+1)= x(t) + \alpha \left(\PiD\left[x(t) - \nabla_{x(t+1)} L^\lambda_P(x(t+1),x(t))\right] -x(t)\right)\ ,
\end{equation}
where $\nabla_{x(t+1)} L_P^\lambda(x(t+1),x(t))\in\R^{2K}$ denotes the gradient of $L_P^\lambda$ with respect to the variable $x(t+1)$ and $\PiD:\R^{2K}\rightarrow \Dcal$ is the projection operator onto $\Dcal$ defined in \eqref{proj_op}.

\begin{lemma} For every $\lambda \in [0,1]$  it holds that $L_P(x(t+1),x(t)) \le L_P^\lambda(x(t+1),x(t))$
and thus, a global minimum $(x^*,x^*)\in\Dcal$ of $L_P^\lambda$ is also a global minimum of $L_P$.
\end{lemma}
\begin{proof}
If $P_A(r_A,Q_A(t+1)) \ne P_A(r_A,Q_A(t))$ or $P_B(r_B,Q_B(t+1)) \ne P_B(r_B,Q_B(t))$, the quantities
\begin{equation*}
P_B(r_B,Q_B(t+1)) - \lambda P_A(r_A,Q_A(t))-(1-\lambda)P_B(r_B,Q_B(t))
\end{equation*}
and
\begin{equation*}
\lambda P_A(r_A,Q_A(t))+(1-\lambda)P_B(r_B,Q_B(t))-P_A(r_A,Q_A(t+1))
\end{equation*}
have opposite signs for every $0\le \lambda\le 1$. Since $\Phi$ is strictly convex with a minimum $\Phi(0)$, we obtain that
\begin{eqnarray*}
\lefteqn{\Phi\left[P_B(r_B,Q_B(t+1)) - P_A(r_A,Q_A(t+1))\right] \le }\nonumber\\
&& \frac{1}{\epsilon}\Phi\left[P_B(r_B,Q_B(t+1)) - \lambda P_A(r_A,Q_A(t))-(1-\lambda)P_B(r_B,Q_B(t))\right]  \\
&& + \frac{1}{\epsilon}\Phi\left[\lambda P_A(r_A,Q_A(t))+(1-\lambda)P_B(r_B,Q_B(t))-P_A(r_A,Q_A(t+1))\right]  \nonumber \ .
\end{eqnarray*}
Therefore, we obtain that
\begin{equation*}
L_P(x(t+1),x(t)) \le L_P^\lambda(x(t+1),x(t))
\end{equation*}
and thus, a global minimum $(x^*,x^*)\in\Dcal$ of $L_P^\lambda$ is also a global minimum of $L_P$.
\end{proof}

The next result is the analogue of lemmas \eqref{uni_sol} and \eqref{Lip_lemma} for the dynamical system \eqref{PGS_lambda}. We skip its proof, since this is identical to those of lemmas \eqref{uni_sol} and \eqref{Lip_lemma}.

\begin{lemma}\label{uni_Lip_sol_lambda} Let $U_{\beta}',\psi_{\beta}',\Phi'$ be Lipschitz continuous and assume that $r_{\beta} \in R_{\beta}$ where $R_{\beta}$ is a bounded subset of $\R_{0}^{+}$. Then,
\begin{itemize}
\item[(i)] The discrete-time dynamical system defined implicitly by \eqref{PGS_lambda} is well posed\footnote{ i.e. to every $x_0\in\Dcal$ corresponds a unique forward trajectory $(x_n)_{n\in\N}$ such that for every $n\in\N$ the pair $(x_{n},x_{n-1})$ satisfies the recurrence \eqref{PGS_lambda}.}.
\item[(ii)] the functional $L_P^\lambda$ has Lipschitz continuous partial derivatives;
\item[(iii)] there exists $\alpha_0>0$ such that for every $0<\alpha<\min\{1,\alpha_0\}$ the map $F:\Dcal\rightarrow\Dcal$ implicitly defined by
$F(x) := (1-\alpha)x + \alpha \PiD [x - \nabla_{1} L_P^\lambda(F(x),x)]$
is Lipschitz continuous.
\end{itemize}
\end{lemma}

The following result provides conditions under which the dynamical system \eqref{PGS_lambda} converges to a fixed point yielding: (i) a unique price for the contingent claim to be traded, and (ii) a global minimum of $L_P^\lambda$ and $L_P$.

\begin{theorem}\label{thm_conv_lambda}  Assume Lipschitz continuity of the functions $U_{\beta}', \psi_{\beta}', \Phi'$, $\beta=A,B$.
Let the agents initial beliefs $Q_{\beta}(0)\in\Delta^K$ be such that $P(r_A,Q_A(0))>P(r_B,Q_B(0))$ for levels of risk $r_{\beta} \in R_{\beta}\subseteq\R_0^+$ and fixed wealths $w_{\beta}\in\R$, $\beta=A,B$.

Then, for every $\epsilon>0$ there exists $\alpha_0:=\alpha_0(\epsilon)>0$ such that for every $0<\alpha<\min\{1,\alpha_0\}$ every forward orbit of the projected gradient scheme \eqref{PGS_lambda} converges to some fixed point  $(Q_A^*,Q_B^*)\in\mathrm{int}(\Dcal)$. Moreover, every fixed point $(Q_A^*,Q_B^*)\in\mathrm{int}(\Dcal)$ is such that 
\begin{equation*}
P_B(r_B,Q_B^*) - P_A(r_A,Q_A^*) = 0 
\end{equation*}
and $(Q_A^*,Q_B^*,Q_A^*,Q_B^*)\in(\Dcal)^2$ is a minimizer of both $L_P$ (given in \eqref{LPDEF}) and $L_P^\lambda$ (given in \eqref{L_pen_lambda}).
\end{theorem}
\begin{proof}
The proof is similar to that of theorem \ref{thm_conv}, so we skip most details. 

If $x^*=(Q_A^*,Q_B^*)\in\mathrm{int}(\Dcal)$ is a fixed point of the projected gradient scheme \eqref{PGS_lambda}, then $\nabla_{1} L_P^\lambda(x^*,x^*)=0$ by proposition \ref{fp_car}. From the definition of the functional $L_P^\lambda$ in \eqref{PGS_lambda}, we get that $(Q_A^*,Q_B^*)\in\mathrm{int}(\Dcal)$ must satisfy the equalities
\begin{eqnarray*}
\Phi'\left(\lambda P_A(r_A,Q_A^*)+(1-\lambda)P_B(r_B,Q_B^*) - P_A(r_A,Q_A^*)\right)\nabla_{Q_A} P_A (r_A,Q_A^*) &=& 0 \nonumber \\
\Phi'\left(P_B(r_B,Q_B^*) - \lambda P_A(r_A,Q_A^*)-(1-\lambda)P_B(r_B,Q_B^*)\right)\nabla_{Q_B} P_B (r_B,Q_B^*) &=& 0 \ ,
\end{eqnarray*}
where $\nabla_{Q_A} P_A$ and $\nabla_{Q_B} P_B$ denote, respectively, the gradient vectors of $P_A$ and $P_B$ with respect to $Q_A$ and $Q_B$. Since the contingent claim $F$ is assumed to be non-constant, we obtain from lemma \ref{P_dif} that such gradient vectors must be non-zero and the equalities above reduce to 
\begin{eqnarray*}
P_A(r_A,Q_A^*) = \lambda P_A(r_A,Q_A^*)+(1-\lambda)P_B(r_B,Q_B^*) = P_B(r_B,Q_B^*) \ .
\end{eqnarray*}
From the set of equalities above, we conclude that $(Q_A^*,Q_B^*,Q_A^*,Q_B^*)\in(\Dcal)^2$ is a minimizer of both functionals $L_P$ and $L_P^\lambda$.
\end{proof}

We will assume that the conditions of theorem \ref{thm_conv_lambda} are satisfied for the remaining of this section. Let $\Delta^+$ denote the set
\begin{equation*}
\Delta^+=\left\{(Q_A^0,Q_B^0)\in\Dcal: P_A(r_A,Q_A^0)\ge P_B(r_B,Q_B^0)\right\} \ ,
\end{equation*} 
and let $\Ecal$ denote the set
\begin{equation*}
\Ecal=\Delta^+\times R_A\times R_B\times[0,1]
\end{equation*}
where $R_A$ and $R_B$ are the bounded sets of lemma \ref{lemma_bounds} (see remark (iii) following this lemma). We define the \emph{asymptotic price} function $P^*:\Ecal\rightarrow\R$ by
\begin{equation}\label{PstarA}
P^*(Q_A^0,Q_B^0,r_A,r_B,\lambda) = \lim_{t\rightarrow \infty} P_A(r_A,Q_A(t)) = P_A(r_A,Q_A^*) \ ,
\end{equation}
where $\{Q_A(t),Q_B(t)\}_{t\in\N}$ is the forward orbit of \eqref{PGS_lambda} with initial condition $(Q_A^0,Q_B^0)\in\Delta^+$ and fixed parameters $r_A$, $r_B$ and $\lambda$. Theorem \ref{thm_conv_lambda} ensures that $P^*$ is well defined by \eqref{PstarA}. Moreover, theorem \ref{thm_conv_lambda} provides the following  alternative representation for $P^*$, which will turn out to be useful below
\begin{equation}\label{PstarB}
P^*(Q_A^0,Q_B^0,r_A,r_B,\lambda) = \lim_{t\rightarrow \infty} P_B(r_B,Q_B(t)) = P_B(r_B,Q_B^*) \ .
\end{equation}

The three results below describe the dependence of the asymptotic price function $P^*$ with respect to the agents initial beliefs $(Q_A^0,Q_B^0)\in\Delta^+$, levels of risk $r_A\in R_A$, $r_B\in R_B$ and relative bargaining power $\lambda\in[0,1]$.

\begin{corollary}
Under the assumptions of theorem \ref{thm_conv_lambda}, the asymptotic price function $P^*$ is continuous on $\Ecal$.
\end{corollary} 
\begin{proof}
Follows from theorem \ref{thm_conv_lambda}, the asymptotic price representations $\eqref{PstarA}$ and $\eqref{PstarB}$, continuity of the agents reservation price functions $P_A$ and $P_B$ with respect to $(r_A,Q_A)\in R_A\times\Delta^K$ and $(r_B,Q_B)\in R_B\times\Delta^K$, respectively, and continuity of the map $F$ of item (iii) in lemma \ref{uni_Lip_sol_lambda}.
\end{proof}

\begin{corollary}
Assume that the conditions of theorem \ref{thm_conv_lambda} hold. Fix the agents initial beliefs $(Q_A^0,Q_B^0)\in\Delta^+$ and levels of risk $r_A\in R_A$ and $r_B\in R_B$. Let $P^*(\lambda)$ denote the asymptotic price function as a function of the agents relative bargaining power $\lambda$. The following statements hold:
\begin{itemize}
\item[(i)] $P^*(\lambda)$ is an increasing function of the relative bargaining power of the two agents $\lambda$;
\item[(ii)] if $\lambda=0$ then $Q_B(t)=Q_B(t+1)$ for every $t\in\N$ and $P^*(\lambda) = P_B(r_B,Q_B(0))$;
\item[(iii)] if $\lambda=1$ then $Q_A(t)=Q_A(t+1)$ for every $t\in\N$ and $P^*(\lambda) = P_A(r_A,Q_A(0))$.
\end{itemize}
\end{corollary} 
\begin{proof}
Follows from theorem \ref{thm_conv_lambda} and the form of the penalty term in \eqref{L_pen_lambda} and \eqref{PGS_lambda}.
\end{proof}

\begin{corollary}
Assume that the conditions of theorem \ref{thm_conv_lambda} hold. Fix the agents initial beliefs $(Q_A^0,Q_B^0)\in\Delta^+$ and relative bargaining power $\lambda$. Let $P^*(r_A,r_B)$ denote the asymptotic price function as a function of the agents levels of risk $r_A\in R_A$ and $r_B\in R_B$. The following statements hold:
\begin{itemize}
\item[(i)] for fixed values of $r_B$, $P^*(r_A,r_B)$ is a decreasing function of $r_A$;
\item[(ii)] for fixed values of $r_A$, $P^*(r_A,r_B)$ is a decreasing function of $r_B$.
\end{itemize}
\end{corollary} 
\begin{proof}
Follows from theorem \ref{thm_conv_lambda}, the asymptotic price representations $\eqref{PstarA}$ and $\eqref{PstarB}$ and lemma \ref{P_dif}.
\end{proof}

\section{The case of exponential utility functions}\label{exp_ut}

In this section we restrict our attention to the case where the two agents have exponential utility functions
\begin{equation}\label{EUF}
U_{\beta}(w) = \frac{1-\exp(-\lambda_{\beta} w)}{\lambda_{\beta}} \ , \qquad \beta=A,B,
\end{equation}
where $\lambda_{\beta}\in\R^+$ are the agents risk aversion coefficients. It is clear that this class of utility function satisfy the Inada conditions. Moreover, as $\lambda_{\beta}\rightarrow 0$ we recover the linear utility function which corresponds to risk neutral agents. From \eqref{res_pr_AB}  we obtain that the reservation prices for agent $\beta$ are given by
\begin{eqnarray}\label{EU_prices}
P_{\beta}(r_{\beta},Q_{\beta}) &=& \ell_{\beta} \frac{1}{\lambda_{\beta}}\ln\left(\frac{E_{Q_{\beta}}[\exp(\ell_{\beta}\lambda_{\beta} F)]}{1+\lambda_{\beta} r_{\beta}\exp(\lambda_{\beta} w_{\beta})}\right),\,\,\,\, \beta=A,B. 
\end{eqnarray}
Note that when we restrict our attention to this particular class of utility functions, the corresponding price functions $P_{\beta}$  are much more regular than what is claimed in lemmas \ref{P_dif} and \ref{DP_Lip}. In this particular setup the price functions are $C^\infty$ functions of $(r_{\beta},Q_{\beta})$ on an open neighborhood of $\R_0^+\times\Delta^K$. Its partial derivatives are given by 
\begin{eqnarray*}
\frac{\partial P_{\beta}}{\partial Q_{\beta}^k}(r_{\beta},Q_{\beta})  = \ell_{\beta} \frac{1}{\lambda_{\beta}}\left(\frac{\exp(\ell_{\beta} \lambda_{\beta} F[k])}{E_{Q_{\beta}}[\exp(\ell_{\beta} \lambda_{\beta} F)]} - \frac{\exp(\lambda_{\beta} w_{\beta})}{1+ \lambda_{\beta} r_{\beta}\exp(\lambda_{\beta} w_{\beta})}\right)\ ,  
\end{eqnarray*}
where $k=1,\ldots, K$, $\beta=A,B$  and 
\begin{eqnarray*}
\frac{\partial P_{\beta}}{\partial r_{\beta}}(r_{\beta},Q_{\beta})  = -\ell_{\beta} \frac{\exp(\lambda_{\beta} w_{\beta})}{1+ \lambda_{\beta}r_{\beta}\exp(\lambda_{\beta} w_{\beta})}, \,\,\, \beta=A,B.
\end{eqnarray*}
We remark that one must use lemma \ref{P_dif} to obtain properly defined partial derivatives of $P_A$ and $P_B$ on a neighborhood of $\R^+\times\Delta^K$. Since $\Delta^K$ is a closed subset of $\R^K$, differentiation of the explicit expression of the reservation prices $P_A$ and $P_B$ with respect to the agents beliefs $Q_A$ and $Q_B$, should either be avoided or an extension of $P_A$ and $P_B$ to a neighborhood of $\Delta^K$ should be carefully chosen to preserve relations \eqref{res_pr_AB}.

\subsection{Beliefs optimization for fixed levels of risk}
Due to the convexity properties of the reservation prices $P_A$ and $P_B$ stated in lemma \ref{P_wd},  theorem \ref{opt_existence} ensures the existence of a unique solution to the minimization problem \eqref{PRIMAL} with fixed levels of risk $r_A$ and $r_B$ when the agents make their decisions using the exponential utility functions $U_A$ and $U_B$ in \eqref{EUF}. Figure \ref{Minimization_fig} provides some geometric intuition for the solution of such minimization problem when the number of states of the world is $K=2$. Given a contingency claim $F$ on $\Omega$ and initial beliefs $Q_A^0,Q_B^0\in\Delta^K$, the tangency point between the minimum level set of $\lambda\psi_A(Q_A,Q_A^0)+(1-\lambda)\psi_B(Q_B,Q_B^0)$ and the manifold $P_A(r_A,Q_A)= P_B(r_B,Q_B)$ determines the solution $(Q_A^*,Q_B^*)\in\Dcal$ to the minimization problem \eqref{PRIMAL} for fixed levels of risk $r_A$ and $r_B$, wealths $w_A$ and $w_B$, and relative bargaining power $\lambda$. 

\begin{figure}[h!]
	\centering
      \psfrag{QA1}[cc][][0.75][0]{$Q_A^1$}%
      \psfrag{QB1}[cc][][0.75][0]{$Q_B^1$}%
      \psfrag{0}[cc][][0.75][0]{$0$}
      \psfrag{0.2}[cc][][0.75][0]{$0.2$}
      \psfrag{0.4}[cc][][0.75][0]{$0.4$}
      \psfrag{0.6}[cc][][0.75][0]{$0.6$}
      \psfrag{0.8}[cc][][0.75][0]{$0.8$}
      \psfrag{1}[cc][][0.75][0]{$1$}
      \psfrag{PAPB}[cc][][0.75][0]{$P_A=P_B$}
      \psfrag{QA0QB0}[cc][][0.75][0]{$(Q_A^0,Q_B^0)$}
      \psfrag{QAs}[cc][][0.75][0]{$Q_A^*$}
      \psfrag{QBs}[cc][][0.75][0]{$Q_B^*$}
      \psfrag{psiApsiBc0}[cb][][0.75][0]{$\lambda\psi_A+(1-\lambda)\psi_B = c_0$}
   		\includegraphics[width=80mm,angle=-90]{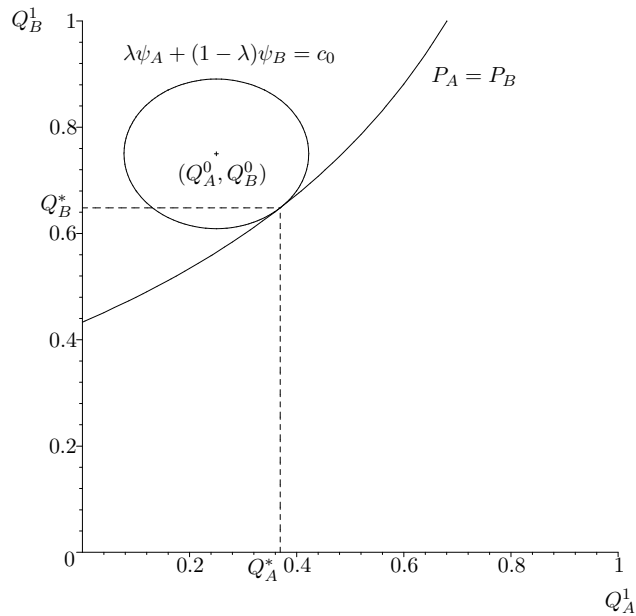}
   		\caption{The optimal solution for the minimization problem \eqref{PRIMAL}. The horizontal axis represents the first component of agent A beliefs' $Q_A=(Q_A^1,Q_A^2)\in\Delta^2$, while the vertical axis represents the first component of agent B beliefs' $Q_B=(Q_B^1,Q_B^2)\in\Delta^2$. The first components of the optimal solution are denoted as $Q_A^*$ and $Q_B^*$. This figure is based on a numerical simulation with the following values of parameters: $K=2$, $F=(1,2)$, $w_A=w_B=0$, $r_A=0.4$, and $r_B=0.3$. The strictly convex ``distance functions'' used are $\psi_A(Q_A,Q_A^0)=\left\|Q_A-Q_A^0\right\|^2/2$ and $\psi_B(Q_B,Q_B^0)=\left\|Q_B-Q_B^0\right\|^2/2$, and the relative bargaining power of the two agents is $\lambda=0.4$.  The agents utility functions are of the form \eqref{EUF} with $\lambda_A=2$ and $\lambda_B=1$. The initial agents beliefs are $Q_A^0=(0.25,0.75)$ and $Q_B^0=(0.75,0.25)$.}\label{Minimization_fig}
\end{figure}

For values of the initial beliefs $Q_A^0,Q_B^0\in\Delta^K$ on the set 
\begin{equation*}
\Delta^+=\left\{(Q_A^0,Q_B^0)\in\Dcal: P_A(r_A,Q_A^0)\ge P_B(r_B,Q_B^0)\right\} \ ,
\end{equation*} 
the minimum level set of the functional $\lambda\psi_A(Q_A,Q_A^0)+(1-\lambda)\psi_B(Q_B,Q_B^0)$ will be contained in $\Delta^+$ and thus, the optimal solution $(Q_A^*,Q_B^*)\in\Dcal$ contained on the manifold $P_A(r_A,Q_A)=P_B(r_A,Q_A)$ will depend continuously on the parameters $\lambda$, $r_A$ and $r_B$, as well as on the initial beliefs in $\Delta^+$.

The unique solution $(Q_A^*,Q_B^*)\in\Dcal$ determines then a unique price for the contingent claim. As seen in Theorem \ref{opt_existence}, such unique optimal price is also a continuous function of the parameters $\lambda$, $r_A$ and $r_B$. Figures \ref{Plambda_fig}, \ref{Prarb_fig3D} and \ref{Prarb_fig2} illustrate this behavior. A few remarks about these figures seem to be appropriate now. 

In what concerns the dependence of the optimal price $P^*$ on the relative bargaining power of the two agents $\lambda$ in Fig. \ref{Plambda_fig}, we have that $P^*$ is an increasing function of $\lambda$. This is consistent with the fact that as $\lambda$ increases, so does the bargaining power of the seller of the contingent claim, thus increasing the value of the optimal price $P^*$. Note also that $\lambda$ measures the ``eccentricity'' of the level sets of the functional $\lambda\psi_A(Q_A,Q_A^0)+(1-\lambda)\psi_B(Q_B,Q_B^0)$ in Fig. \ref{Minimization_fig}.

\begin{figure}[h!]
	\centering
      \psfrag{l}[cc][][0.75][0]{$\lambda$}%
      \psfrag{Ps}[cc][][0.75][0]{$P^*(\lambda)$}%
      \psfrag{0}[ct][][0.75][0]{$0$}
      \psfrag{0.2}[ct][][0.75][0]{$0.2$}
      \psfrag{0.4}[ct][][0.75][0]{$0.4$}
      \psfrag{0.6}[ct][][0.75][0]{$0.6$}
      \psfrag{0.8}[ct][][0.75][0]{$0.8$}
      \psfrag{1}[ct][][0.75][0]{$1$}
      \psfrag{1.45}[rc][][0.75][0]{$1.45$}
      \psfrag{1.5}[rc][][0.75][0]{$1.5$}
      \psfrag{1.55}[rc][][0.75][0]{$1.55$}
   		\includegraphics[width=60mm,angle=-90]{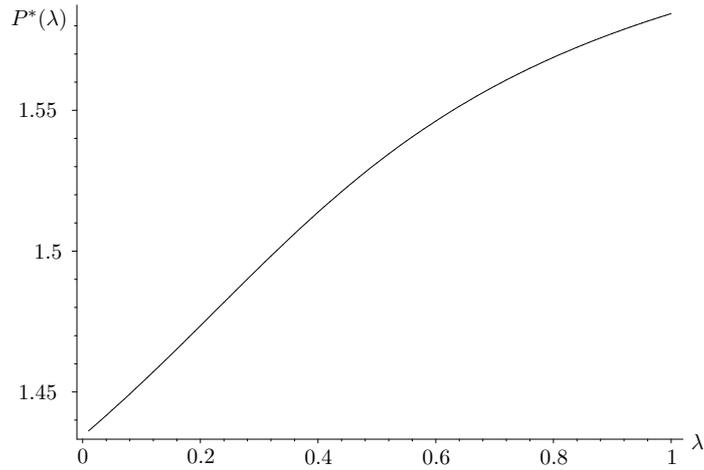}
   		\caption{The optimal price $P^*(\lambda)$ as a function of the relative bargaining power of the two agents $\lambda$, obtained from the minimization problem \eqref{PRIMAL}. The values of the remaining parameters are as in Fig. \ref{Minimization_fig}.}\label{Plambda_fig}
\end{figure}

In what concerns the dependence of the optimal price function $P^*$ on the agents levels of risk $r_A$ and $r_B$, we first recall that its domain of definition is bounded by lemma \ref{lemma_bounds} (see remark (iii) following this lemma), with the upper bounds on $r_A$ and $r_B$ determined by condition \eqref{dif_ine}. It is particularly interesting to note that $P^*$ is not a monotone function of $r_A$ for small values of $r_B$ (see Fig. \ref{Prarb_fig3D} and Fig. \ref{PrA_rBeq0_fig}). If agent B level of risk $r_B$ is sufficiently small, even though agent B bargaining power is larger than agent A bargaining power, agent A is able to increase the value of the optimal price $P^*$ by taking positive values of risk up to a given amount, from where the optimal price starts to decrease with agent A level of risk $r_A$. For larger values of $r_B$, $P^*$ is a strictly decreasing function of $r_A$ (see Figs. \ref{Prarb_fig3D} and  \ref{PrA_rBeq02_fig}) and $P^*$ is always an increasing function of $r_B$ for fixed values of $r_A$ (see Figs. \ref{Prarb_fig3D}, \ref{PrB_rAeq0_fig} and \ref{PrB_rAeq02_fig}).

\begin{figure}[h!]
	\centering
      \psfrag{rA}[ct][][0.75][0]{$r_A$}%
      \psfrag{rB}[ct][][0.75][0]{$r_B$}%
      \psfrag{Ps}[rc][][0.75][0]{$P^*(r_A,r_B)$}%
      \psfrag{0}[ct][][0.75][0]{$0$}
      \psfrag{0.1}[ct][][0.75][0]{$0.1$}
      \psfrag{0.2}[ct][][0.75][0]{$0.2$}
      \psfrag{0.3}[ct][][0.75][0]{$0.3$}
      \psfrag{0.4}[ct][][0.75][0]{$0.4$}
      \psfrag{1.4}[rc][][0.75][0]{$1.4$}
      \psfrag{1.5}[rc][][0.75][0]{$1.5$}
      \psfrag{1.6}[rc][][0.75][0]{$1.6$}
      \psfrag{1.7}[rc][][0.75][0]{$1.7$}
   		\includegraphics[width=60mm,angle=-90]{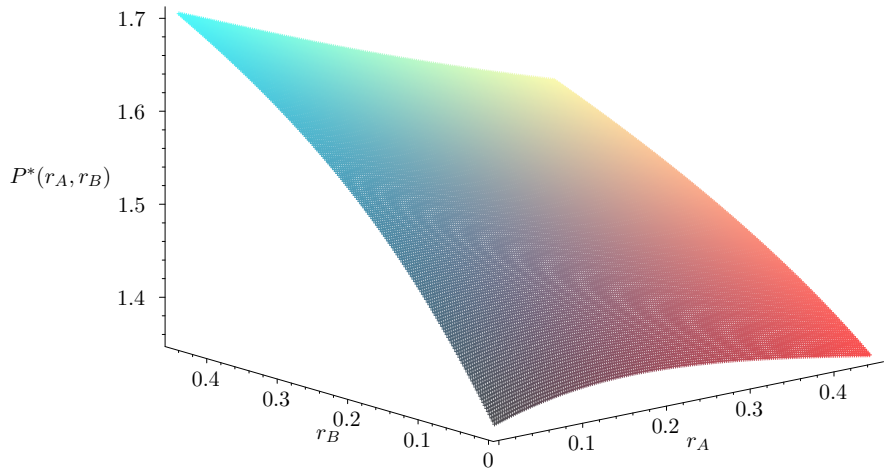}
   		\caption{The optimal price $P^*(r_A,r_B)$ as a function of the agents levels of risk $r_A$ and $r_B$, obtained from the minimization problem \eqref{PRIMAL}. The values of the remaining parameters are as in Fig. \ref{Minimization_fig}.}\label{Prarb_fig3D}
\end{figure}

We should remark that the typical situation is the one observed in Figs. \ref{PrA_rBeq02_fig}, \ref{PrB_rAeq0_fig} and \ref{PrB_rAeq02_fig}, i.e. if all other parameters are kept fixed, $P^*$ decreases with the level of risk of the seller of the contingent claim and increases with the level of risk of the buyer of the contingent claim. The behavior of Fig. \ref{PrA_rBeq0_fig} is observed only for a small set of suitably chosen parameters.

\begin{figure}[h!]
	\centering
      \psfrag{PsrA}[rc][][0.65][0]{$P^*(r_A)$}%
      \psfrag{PsrB}[rb][][0.65][0]{$P^*(r_B)$}%
      \psfrag{rA}[cc][][0.65][0]{$r_A$}%
      \psfrag{rB}[cc][][0.65][0]{$r_B$}%
      \psfrag{0}[ct][][0.65][0]{$0$}
      \psfrag{0.1}[ct][][0.65][0]{$0.1$}
      \psfrag{0.2}[ct][][0.65][0]{$0.2$}
      \psfrag{0.3}[ct][][0.65][0]{$0.3$}
      \psfrag{0.4}[ct][][0.65][0]{$0.4$}
      \psfrag{1.36}[rc][][0.65][0]{$1.36$}
      \psfrag{1.37}[rc][][0.65][0]{$1.37$}
      \psfrag{1.38}[rc][][0.65][0]{$1.38$}
      \psfrag{1.4}[rc][][0.65][0]{$1.4$}
      \psfrag{1.46}[rc][][0.65][0]{$1.46$}
      \psfrag{1.5}[rc][][0.65][0]{$1.5$}
      \psfrag{1.54}[rc][][0.65][0]{$1.54$}
      \psfrag{1.6}[rc][][0.65][0]{$1.6$}
      \psfrag{1.7}[rc][][0.65][0]{$1.7$}
   		\subfigure[][]{\includegraphics[width=40mm,angle=-90]{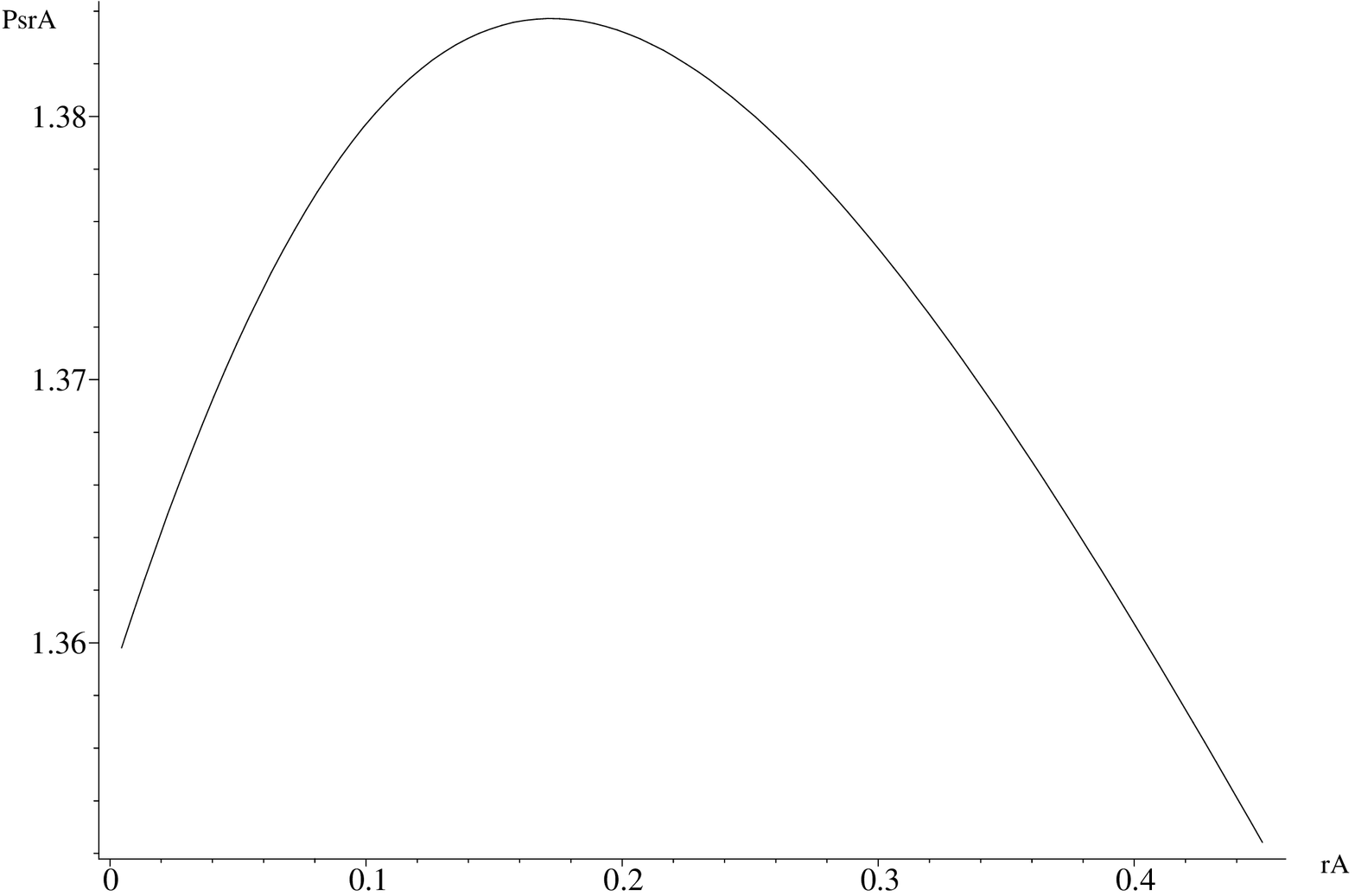}\label{PrA_rBeq0_fig}} \qquad
   		\subfigure[][]{\includegraphics[width=40mm,angle=-90]{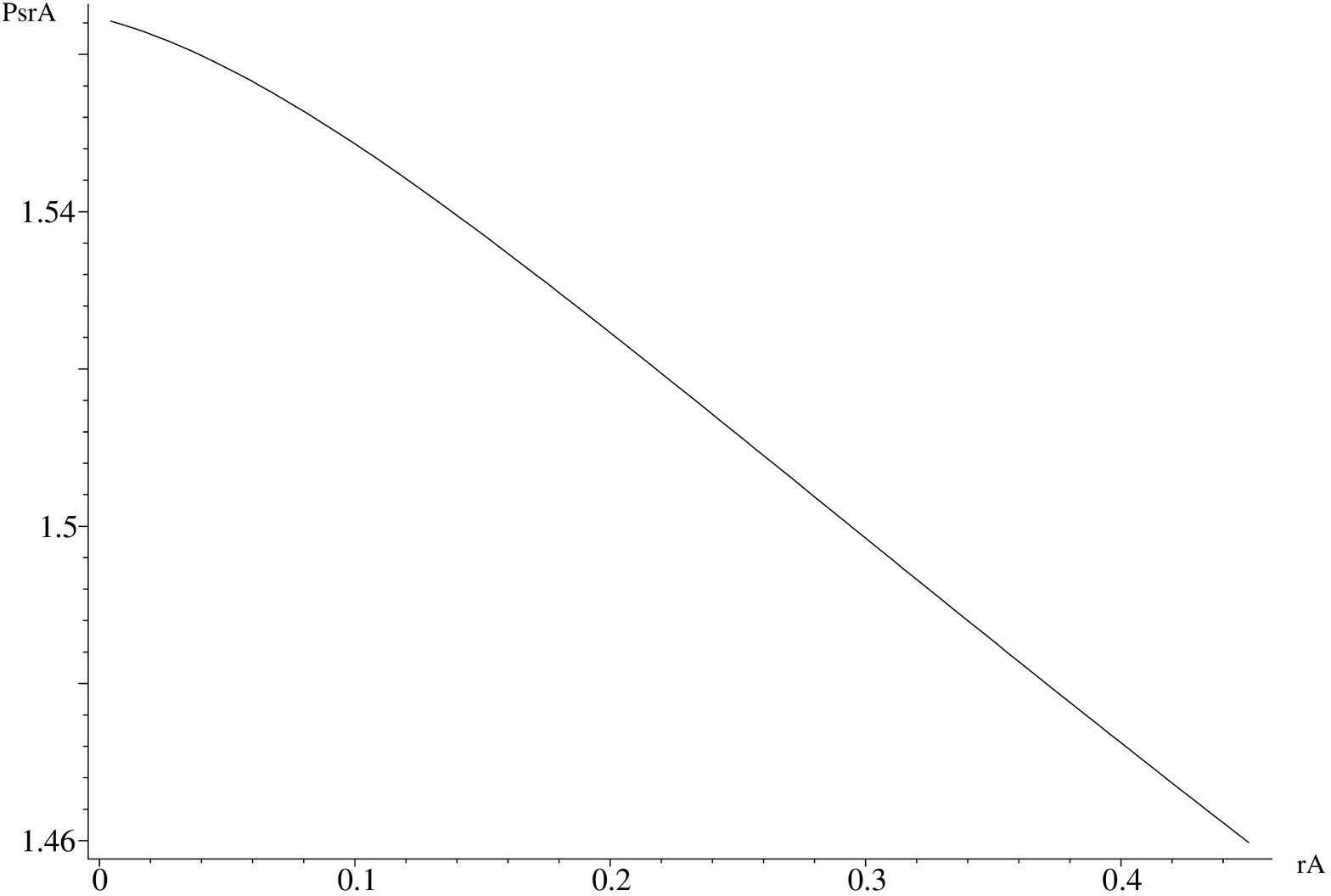}\label{PrA_rBeq02_fig}} \\
   		\subfigure[][]{\includegraphics[width=40mm,angle=-90]{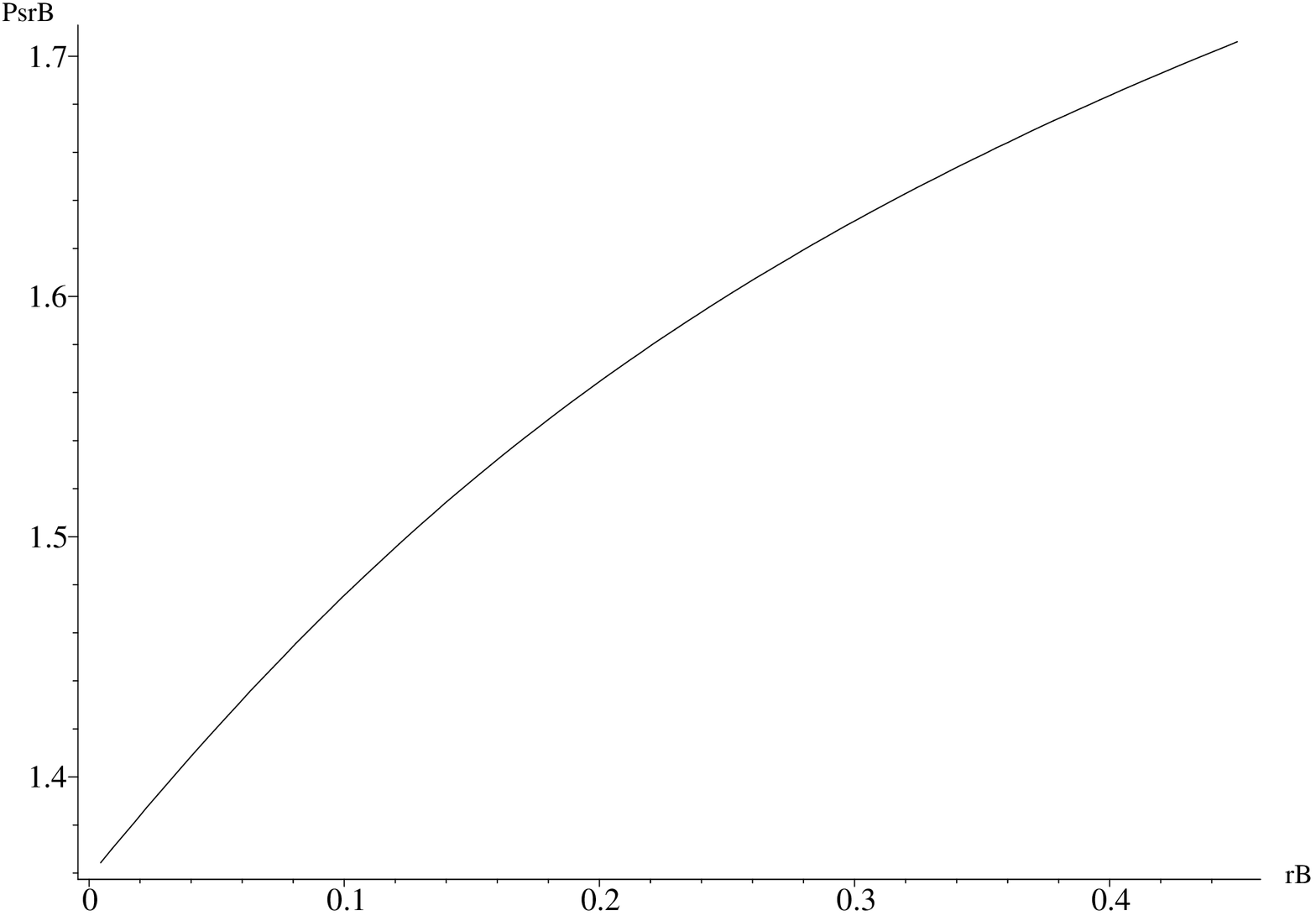}\label{PrB_rAeq0_fig}} \qquad
   		\subfigure[][]{\includegraphics[width=40mm,angle=-90]{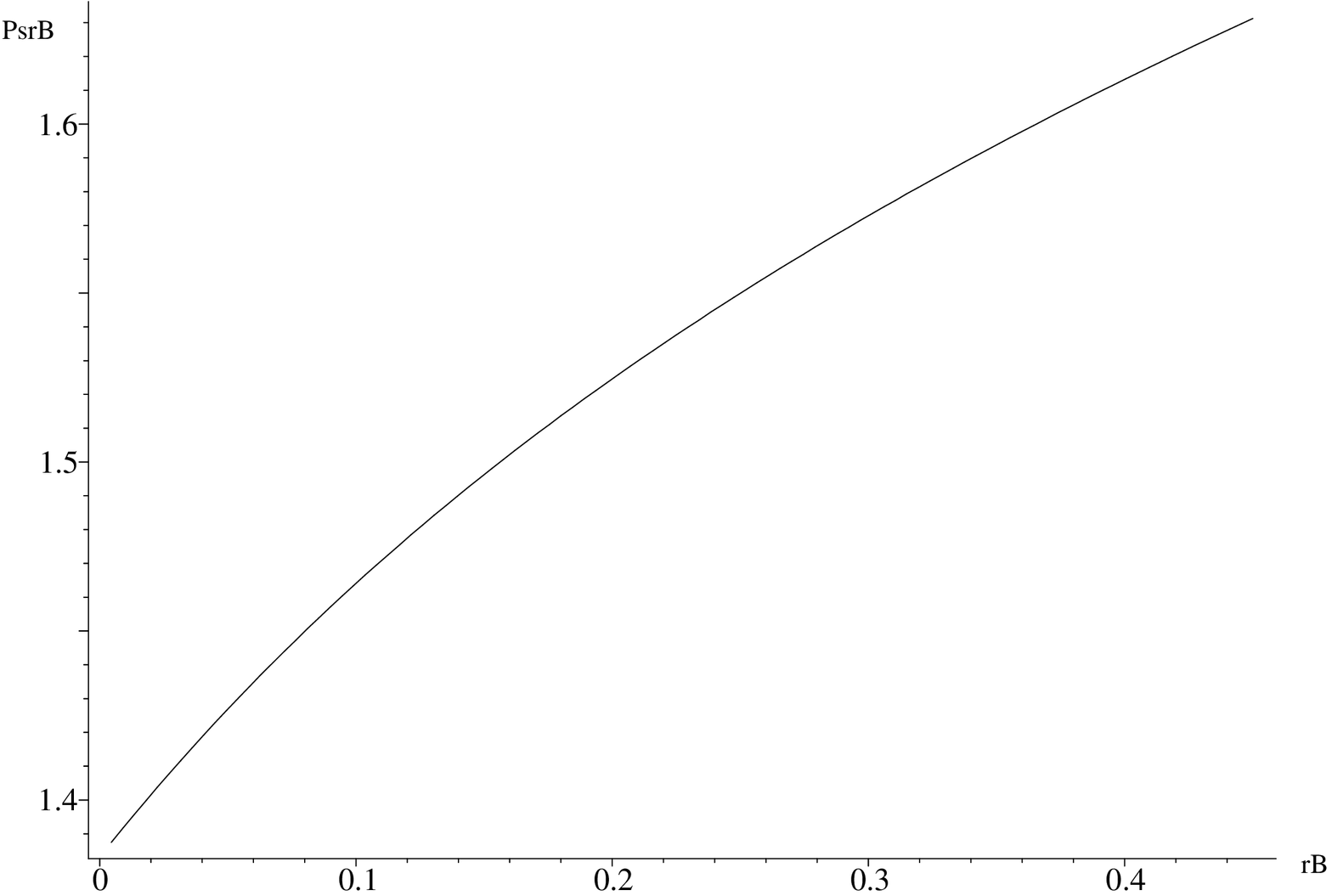}\label{PrB_rAeq02_fig}} 
   		\caption{Sections of the three-dimensional graph of $P^*(r_A,r_B)$ in Fig. \ref{Prarb_fig3D}. The figures on the top are plots of the optimal price $P^*(r_A)$ as a function of agent A level of risk $r_A$ when agent B level of risk is fixed at $r_B=0$ (left) and $r_B=0.2$ (right). The figures on the bottom are plots of the optimal price $P^*(r_B)$ as a function of agent B level of risk $r_B$ when agent A level of risk is fixed at $r_A=0$ (left) and $r_A=0.2$ (right).}\label{Prarb_fig2}
\end{figure}

\subsection{The bargaining dynamics with exponential utility functions} We will now focus our attention on the behavior of the projected dynamical systems \eqref{PGS} and \eqref{PGS_lambda} modeling the bargaining process under which the two agents exchange beliefs and eventually agree on a unique price to trade the contingent claim.

We start with projected dynamical systems \eqref{PGS}. In Fig. \ref{IPGS_fig} we plot the evolution of the two agents reservation prices $P_A(t)$ and $P_B(t)$ determined by \eqref{EU_prices} evaluated along orbits 
\begin{equation*}
(Q_A(t),Q_B(t))_{t\ge 0}
\end{equation*}
of \eqref{PGS} for three different values of the agents relative bargaining power $\lambda$. We note that even though \eqref{PGS} achieves its primary goal of modeling an exchange of beliefs under which the agents agree on a unique price for the contingent claim, we believe it to have several flaws, namely:
\begin{itemize}
\item[(i)] the dynamics of \eqref{PGS} depend very weakly on the agents relative bargaining power $\lambda$, which can immediately be seen from a comparison of Figs. \ref{IPGS_l0_fig}, \ref{IPGS_l1_fig} and \ref{IPGS_l04_fig};
\item[(ii)] the reservation price of the seller is not necessarily a decreasing function of time $t$;
\item[(iii)] the reservation price of the buyer is not necessarily an increasing function of time $t$;
\item[(iv)] the asymptotic price does not need to lie on the interval bounded by $P_A(0)$ and $P_B(0)$;
\item[(v)] for extreme values of the agents relative bargaining power $\lambda=0$ and $\lambda=1$, the functions $P_B(t)$ and $P_A(t)$ do not need to be constant with $t$.
\end{itemize}
The main reason for the projected gradient system \eqref{PGS} to have the qualitative behavior described in items (i)-(v) above is the lack of dependence of the penalty function in \eqref{LPDEF} on the the agents relative bargaining power $\lambda$, as well as the agents reservation prices at the previous period of time. 

\begin{figure}[h!]
	\centering
      \psfrag{PAt}[cc][][0.75][0]{$P_A(t)$}%
      \psfrag{PBt}[cc][][0.75][0]{$P_B(t)$}%
      \psfrag{t}[lc][][0.75][0]{$t$}%
      \psfrag{0}[ct][][0.65][0]{$0$}
      \psfrag{10}[ct][][0.65][0]{$10$}
      \psfrag{20}[ct][][0.65][0]{$20$}
      \psfrag{30}[ct][][0.65][0]{$30$}
      \psfrag{40}[ct][][0.65][0]{$40$}
      \psfrag{50}[ct][][0.65][0]{$50$}
      \psfrag{1.2}[rc][][0.65][0]{$1.2$}
      \psfrag{1.3}[rc][][0.65][0]{$1.3$}
      \psfrag{1.4}[rc][][0.65][0]{$1.4$}
      \psfrag{1.5}[rc][][0.65][0]{$1.5$}
   		\subfigure[][]{\includegraphics[width=40mm,angle=-90]{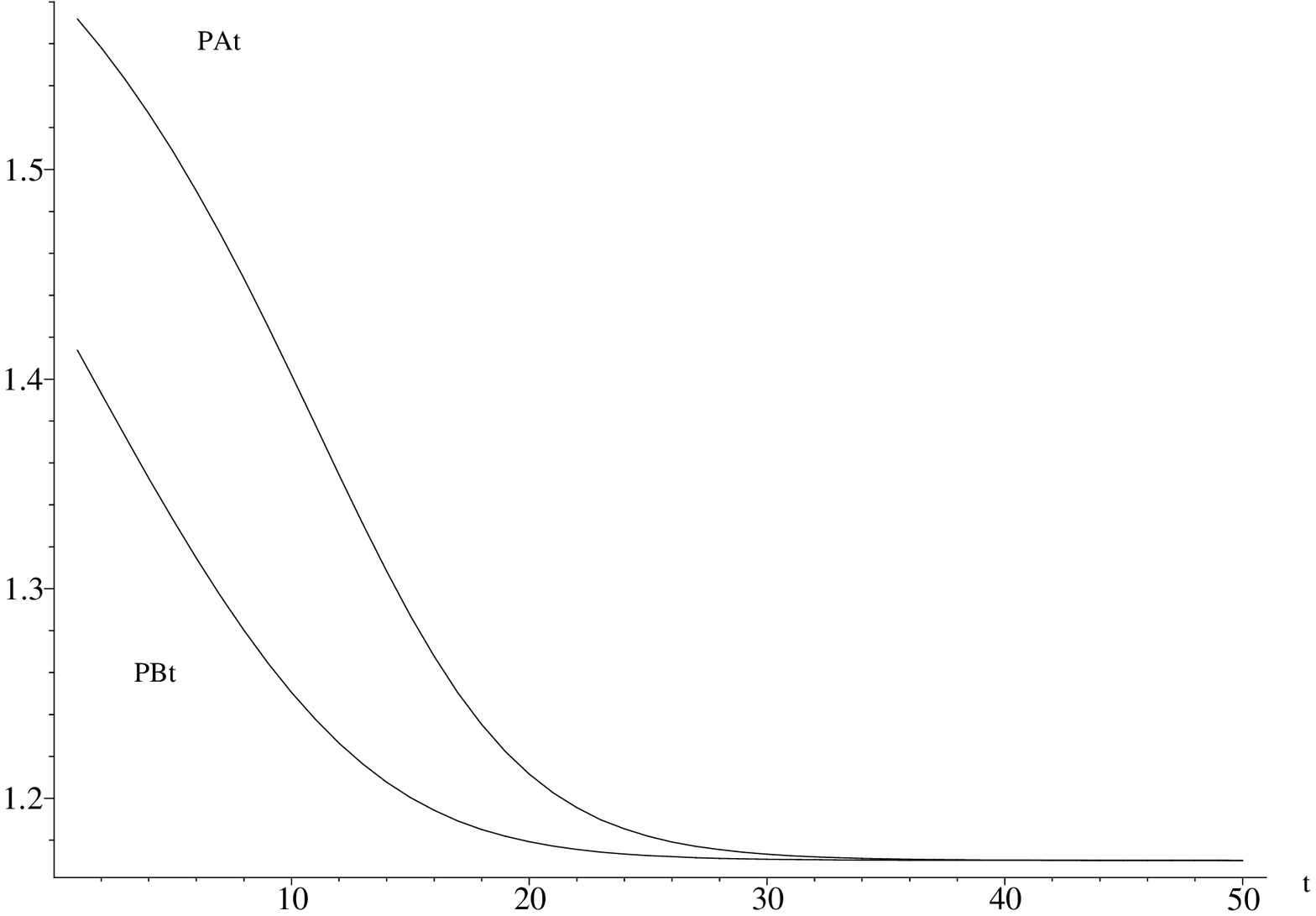}\label{IPGS_l0_fig}} \qquad
   		\subfigure[][]{\includegraphics[width=40mm,angle=-90]{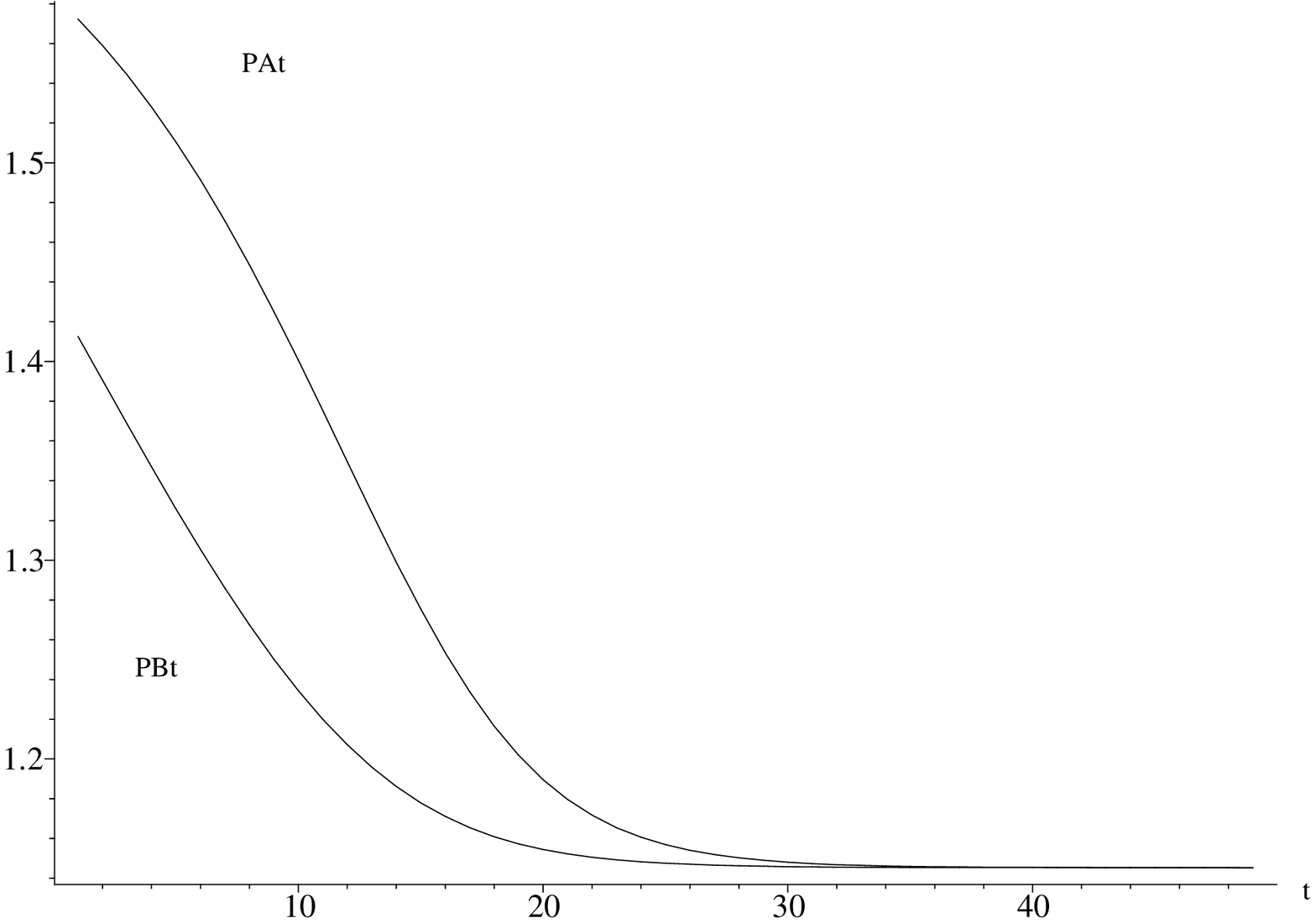}\label{IPGS_l1_fig}} \\
   		\subfigure[][]{\includegraphics[width=40mm,angle=-90]{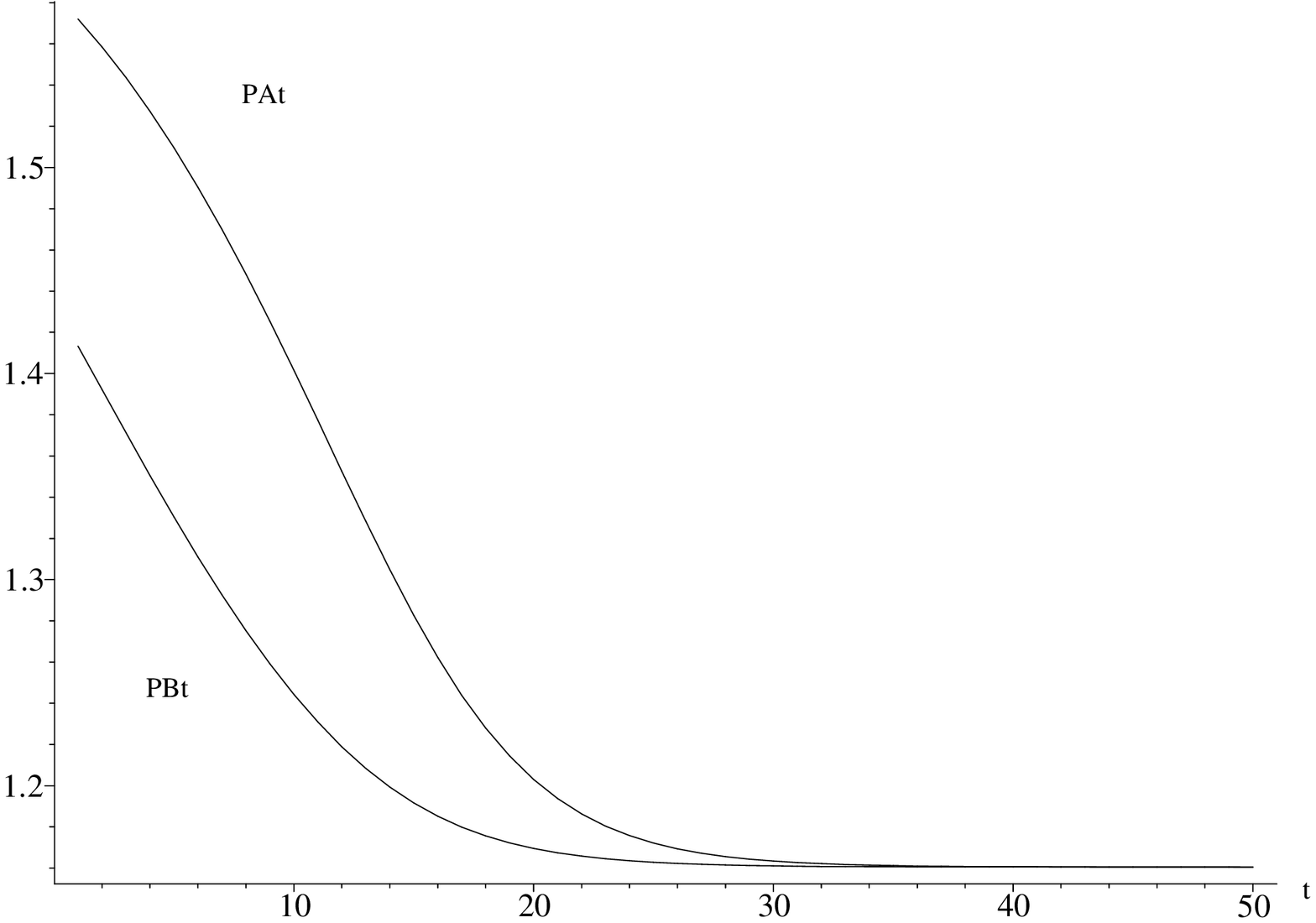}\label{IPGS_l04_fig}} \quad
   		   		\caption{The reservation price functions of both agents computed along orbits of the projected gradient system \eqref{PGS} modeling the bargaining process for different values of the relative bargaining power of the two agents $\lambda$: $\lambda=0$ on the top left, $\lambda=1$ on the top right, and $\lambda=0.4$ on the bottom. For all three figures $\epsilon=0.1$ and $\alpha=0.05$. The remaining values of parameters are as in Fig. \ref{Minimization_fig}.}\label{IPGS_fig}
\end{figure}

As described in section \ref{SSBD2}, the introduction of the projected gradient system \eqref{PGS_lambda} addresses the problems of \eqref{PGS} listed in items (i)-(v) above. Therefore, we believe that \eqref{PGS_lambda} may be a particularly good dynamical model for the bargaining process modeling the exchange of beliefs between two agents who are willing to trade a given asset. In particular, the introduction of \eqref{PGS_lambda} provides us with a dynamical system that, under the reasonable assumptions of theorem \ref{thm_conv_lambda}, satisfies the following properties:
\begin{itemize}
\item[1)] the asymptotic price $P^*$ is such that the inequalities $P_B(0)\le P^*\le P_A(0)$ (see Fig. \ref{DPGS_fig});
\item[2)] the asymptotic price $P^*$ increases with $\lambda$ (see Fig. \ref{DPGS_fig});
\item[3)] if $\lambda=0$ then $P_B(t)$ is constant in $t\in\Z$ and $P^*=P_B(r_B,Q_B(0))$ (see Fig. \ref{DPGS_l0_fig}); 
\item[4)] if $\lambda=1$ then $P_A(t)$ is constant in $t\in\Z$ and $P^*=P_A(r_A,Q_A(0))$ (see Fig. \ref{DPGS_l1_fig});
\item[5)] all other parameters being fixed, the asymptotic price $P^*$ is a decreasing function of $r_A$;
\item[6)] all other parameters being fixed, the asymptotic price $P^*$ is an increasing function of $r_B$;
\item[7)] $P_A(t)$ is a decreasing function of $t$ and $P_B(t)$ is an increasing function of $t$ (see Fig. \ref{DPGS_fig}).
\end{itemize}

\begin{figure}[h!]
	\centering
      \psfrag{PAt}[cc][][0.75][0]{$P_A(t)$}%
      \psfrag{PBt}[cc][][0.75][0]{$P_B(t)$}%
      \psfrag{t}[lc][][0.75][0]{$t$}%
      \psfrag{0}[ct][][0.65][0]{$0$}
      \psfrag{20}[ct][][0.65][0]{$20$}
      \psfrag{40}[ct][][0.65][0]{$40$}
      \psfrag{60}[ct][][0.65][0]{$60$}
      \psfrag{80}[ct][][0.65][0]{$80$}
      \psfrag{1.44}[rc][][0.65][0]{$1.44$}
      \psfrag{1.48}[rc][][0.65][0]{$1.48$}
      \psfrag{1.46}[rc][][0.65][0]{$1.46$}
      \psfrag{1.5}[rc][][0.65][0]{$1.5$}
      \psfrag{1.52}[rc][][0.65][0]{$1.52$}
      \psfrag{1.54}[rc][][0.65][0]{$1.54$}
      \psfrag{1.56}[rc][][0.65][0]{$1.56$}
      \psfrag{1.58}[rc][][0.65][0]{$1.58$}
   		\subfigure[][]{\includegraphics[width=40mm,angle=-90]{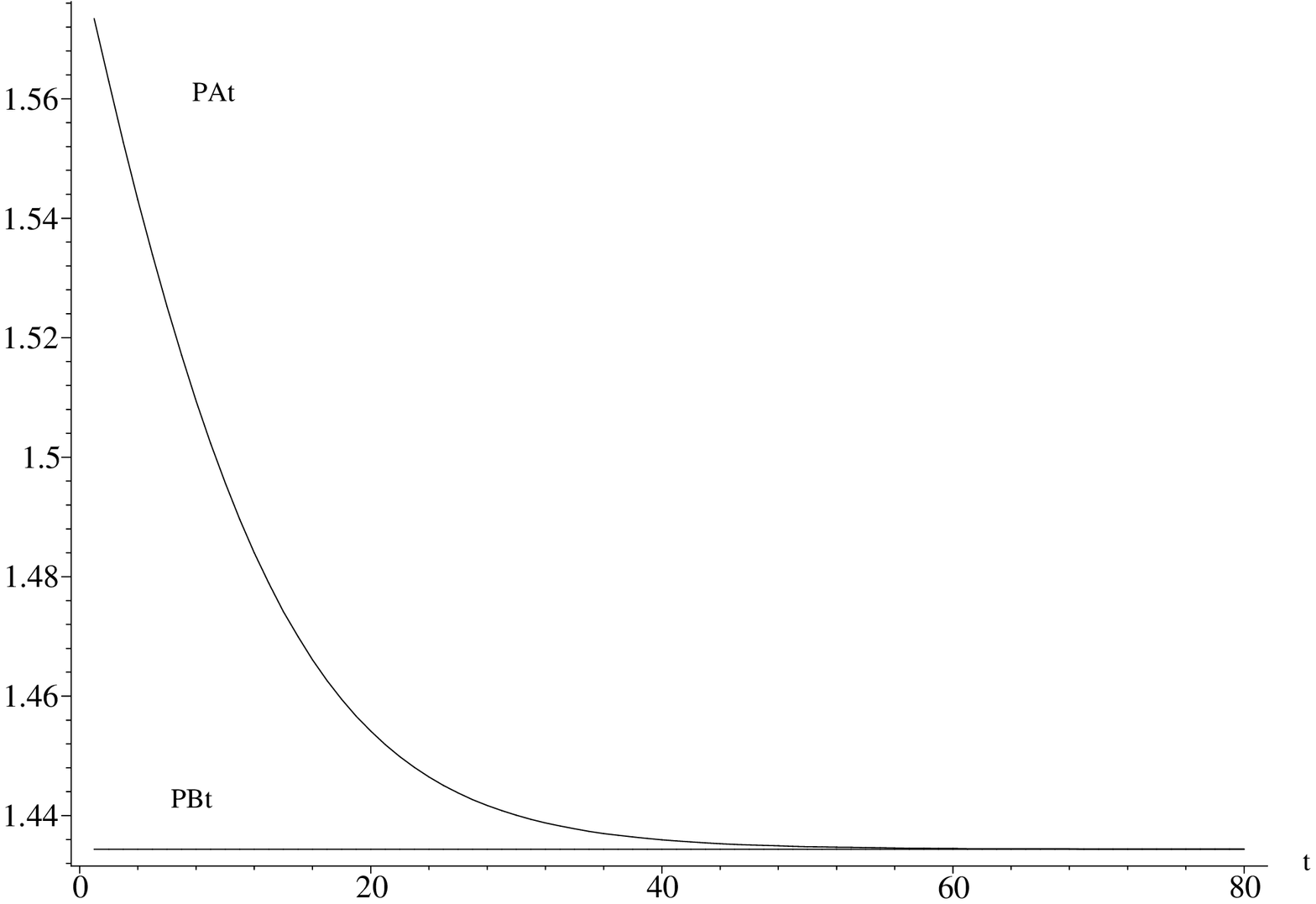}\label{DPGS_l0_fig}} \qquad
   		\subfigure[][]{\includegraphics[width=40mm,angle=-90]{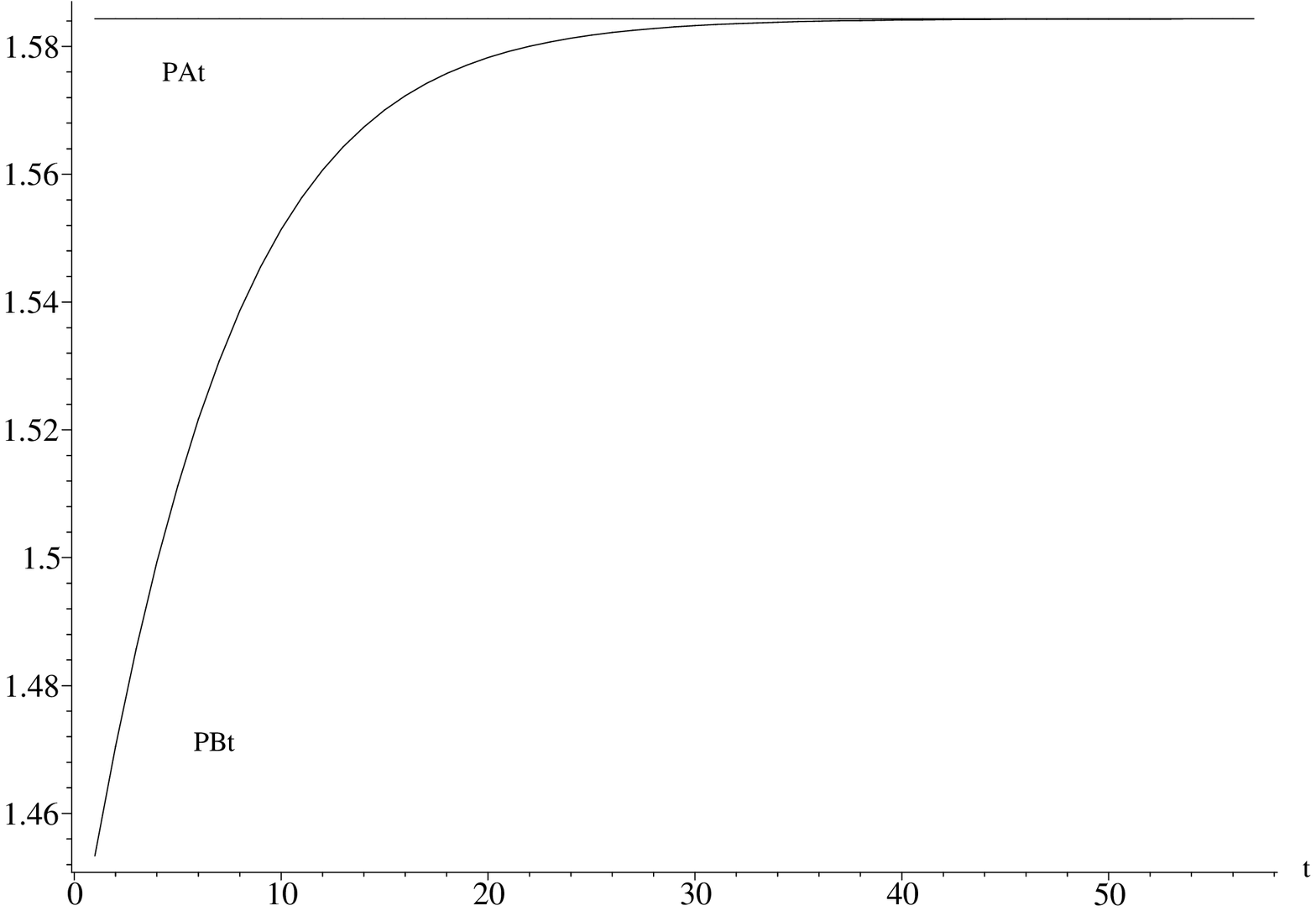}\label{DPGS_l1_fig}} \\
   		\subfigure[][]{\includegraphics[width=40mm,angle=-90]{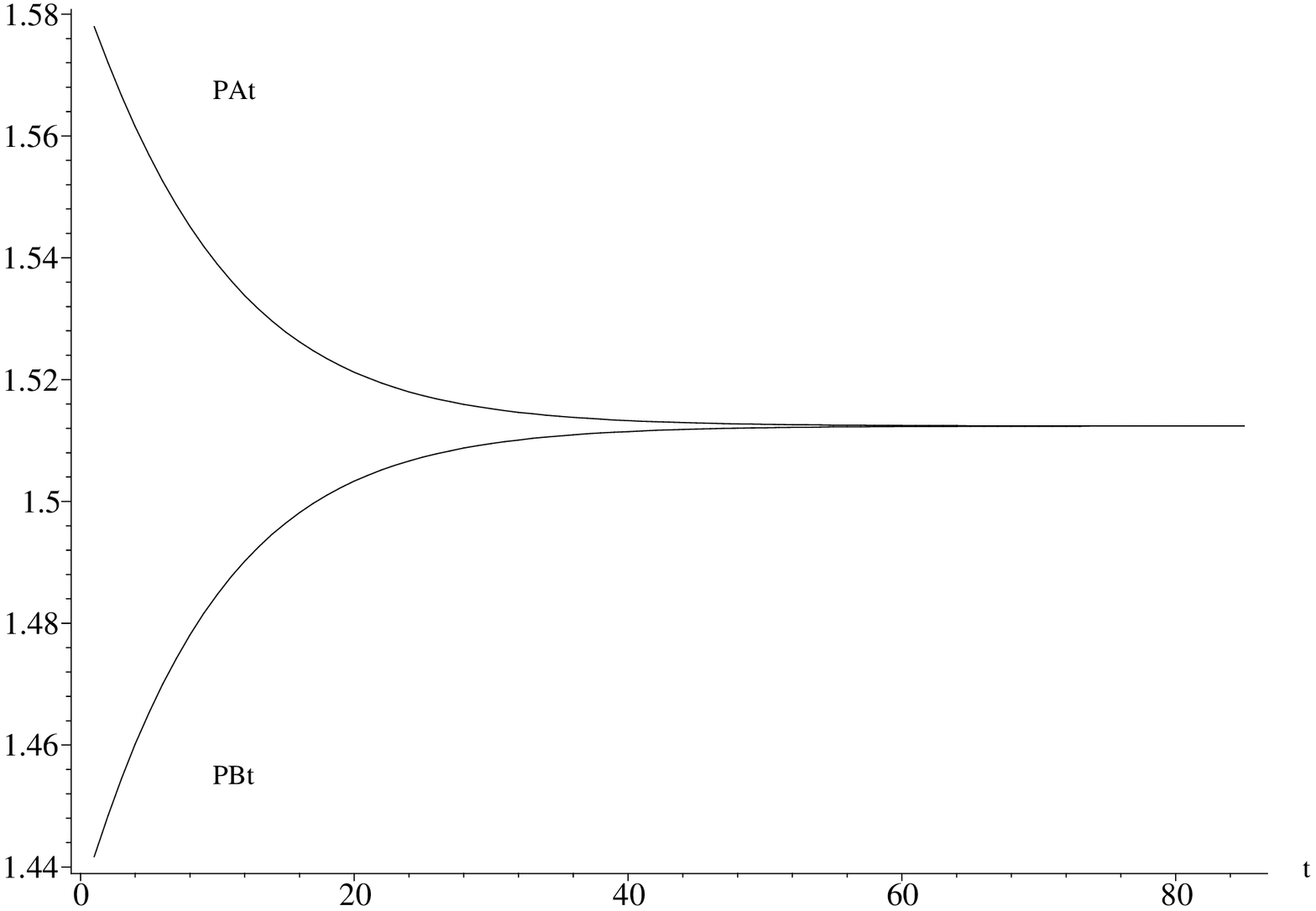}\label{DPGS_l04_fig}} 
   		   		\caption{The reservation price functions of both agents computed along orbits of the projected gradient system \eqref{PGS_lambda} modeling the bargaining process for different values of the relative bargaining power of the two agents $\lambda$: $\lambda=0$ on the top left, $\lambda=1$ on the top right, and $\lambda=0.4$ on the bottom. For all three figures $\epsilon=0.1$ and $\alpha=0.05$. The remaining values of parameters are as in Fig. \ref{Minimization_fig}.}\label{DPGS_fig}
\end{figure}

\section{Conclusion}
We have described a class of dynamical systems modeling a ``bargaining period'' in which the buyer and the seller of a contingent claim, whose real value depends on the future state of the economy, agree on a common price for its transaction. Such dynamical systems model the evolution with time of the agents' beliefs about the future states of the world, i.e. probability measures representing each agent estimate for the occurrence of each possible state of the world. Moreover, this class of dynamical systems has the form of a discrete-time gradient dynamical system, which is related to the minimization problem of the joint belief deviation function. Our gradient dynamical system is shown to converge to a common price for the asset under mild assumptions.

\appendix

\section{Proof of Lemma \ref{P_wd}}

\begin{proof}
(1) We prove the statement for the price function of agent A, the proof for agent B being similar.
Recall that the expected utility function of agent A for the contingent claim $F$, which we will denote by $\overline{U}_A$, is given by 
\begin{equation*}
\overline{U}_A(x) = E_{Q_A}[U_A(w_A + x - F)] = \sum_{k=1}^K Q_A^k U_A(w_A + x - F[k]) \ ,
\end{equation*}
where $Q_A=(Q_A^1,\ldots,Q_A^K)\in\Delta^K$. Since $U_A(w_A + x - F[k])$ is a strictly increasing concave function of $x$ for every $k\in\{1,\ldots,K\}$, we get that so is $\overline{U}_A(x)$, being a (convex) linear combination of such functions. Thus, the expected utility $\overline{U}_A$ is invertible. Therefore, the reservation price $P_A$ of agent A is well-defined as the (unique) solution of the equality
\begin{equation}\label{P_der_1}
U(w_A)-r_A = \overline{U}_A(P_A(Q_A)) \ ,
\end{equation}
for each $(r_A,Q_A)\in\R_0^+\times\Delta^K$.

(2) We will only prove the case $\beta=A$, the proof of the case $\beta=B$ being similar. Let $Q_A\in\Delta^K$ be fixed, and let $r_A^1,r_A^2\in\R_0^+$ and $0\le\nu\le 1$. To make the notation lighter, we introduce the following definitions
\begin{equation*}
P_A^1:=P_A(r_A^1) \ , \quad P_A^2:=P_A(r_A^2) \ , \quad P_A^\nu := P_A(\nu r_A^1+(1-\nu)r_A^2) \ .
\end{equation*}
From \eqref{res_pr_AB} we obtain that
\begin{eqnarray*}
U_A(w_A) -r_A^i &=& E_{Q_A}	\left [U_A(w_A + P_A^i - F)\right ] \ , \qquad i=1,2 \\
U_A(w_A) -\nu r_A^1-(1-\nu)r_A^2 &=& E_{Q_A}	\left [U_A(w_A + P_A^\nu - F)\right ] \ .
\end{eqnarray*}
Using the equalities above and the expected value linearity, we get that
\begin{equation*}
E_{Q_A}	\left [U_A(w_A + P_A^\nu - F)\right ] =  E_{Q_A}	\left [\nu U_A(w_A + P_A^1 - F) + (1-\nu) U_A(w_A + P_A^2 - F)\right ]  \ .
\end{equation*}
Since the equality above is satisfied for every $Q_A\in\Delta^K$, we obtain that
\begin{equation}\label{LC_1}
U_A(w_A + P_A^\nu - F) =  \nu U_A(w_A + P_A^1 - F) + (1-\nu) U_A(w_A + P_A^2 - F)  \ .
\end{equation}
Since $U_A$ is a concave function, we get the inequality
\begin{eqnarray}\label{LC_2}
\lefteqn{\nu U_A(w_A + P_A^1 - F) + (1-\nu) U_A(w_A + P_A^2 - F)}  \nonumber \\
&\le& U_A(\nu(w_A + P_A^1 - F) + (1-\nu)(w_A + P_A^2 - F)) \\
&=& U_A(w_A + \nu P_A^1 +(1-\nu)P_A^2 - F)  \nonumber \ .
\end{eqnarray}
Putting together conditions \eqref{LC_1} and \eqref{LC_2}, we obtain that
\begin{equation*}
U_A(w_A + P_A^\nu - F) \le U_A(w_A + \nu P_A^1 +(1-\nu)P_A^2 - F) \ .
\end{equation*}
Recalling that $U_A$ has positive first derivative, from the previous inequality we get 
\begin{equation*}
P_A^\nu \le \nu P_A^1 +(1-\nu)P_A^2 \ ,
\end{equation*}
which concludes the proof of convexity of agent's A price function $P_A$ with respect to its level of risk $r_A$.

(3) We will prove item (i), the proof of item (ii) being similar. We let $r_A\in\R_0^+$ be fixed, $Q_A^1,Q_A^2\in\Delta^K$ and $0\le\nu\le 1$. We introduce the following simplification to the  notation
\begin{equation*}
P_A^1=P_A(Q_A^1) \ , \quad P_A^2=P_A(Q_A^2) \ , \quad P_A^\nu = P_A(\nu Q_A^1+(1-\nu)Q_A^2) \ .
\end{equation*}
From equality \eqref{res_pr_AB} we get that
\begin{eqnarray*}
U_A(w_A) -r_A &=& E_{Q_A^i}	\left [U_A(w_A + P_A^i - F)\right ] \ , \qquad i=1,2 \\
U_A(w_A) -r_A &=& E_{\nu Q_A^1+(1-\nu)Q_A^2}	\left [U_A(w_A + P_A^\nu - F)\right ] \ .
\end{eqnarray*}
From the equalities above, we immediately obtain that
\begin{eqnarray*}
\lefteqn{E_{\nu Q_A^1+(1-\nu)Q_A^2}	\left [U_A(w_A + P_A^\nu - F)\right ]}\nonumber \\
&=& \nu E_{Q_A^1}	\left [U_A(w_A + P_A^1 - F)\right ] + (1-\nu) E_{Q_A^2}	\left [U_A(w_A + P_A^2 - F)\right ] 
\end{eqnarray*}
which is clearly equivalent to 
\begin{eqnarray}\label{LC_6}
\lefteqn{E_{\nu Q_A^1+(1-\nu)Q_A^2}	\left [U_A(w_A + P_A^\nu - F)\right ]}\nonumber \\
&=& E_{\nu Q_A^1+(1-\nu)Q_A^2}	\left [U_A(w_A + \nu P_A^1 + (1-\nu)P_A^2 - F)\right ] + G_{Q_A^1,Q_A^2}(\nu) \ ,
\end{eqnarray}
where $G_{Q_A^1,Q_A^2}(\lambda)$ is given by
\begin{eqnarray*}
G_{Q_A^1,Q_A^2}(\nu) & = & \nu E_{Q_A^1}	\left [U_A(w_A + P_A^1 - F)\right ] + (1-\nu) E_{Q_A^2}	\left [U_A(w_A + P_A^2 - F)\right ] \nonumber \\
&& - E_{\nu Q_A^1+(1-\nu)Q_A^2}	\left [U_A(w_A + \nu P_A^1 + (1-\nu)P_A^2 - F)\right ]  \ . \nonumber
\end{eqnarray*}
Assume for a moment that $G_{Q_A^1,Q_A^2}(\nu)$ is non-negative for every $Q_1,Q_2\in\Delta^K$ and every $\nu\in[0,1]$. Then, we obtain from \eqref{LC_6} that
\begin{eqnarray*}
\lefteqn{E_{\nu Q_A^1+(1-\nu)Q_A^2}	\left [U_A(w_A + P_A^\nu - F)\right ] }\nonumber \\
&\ge & E_{\nu Q_A^1+(1-\nu) Q_A^2}	\left [ U_A(w_A + \nu P_A^1 + (1-\nu)P_A^2 - F)\right ] \nonumber \\
\end{eqnarray*}
Since the previous inequality holds for arbitrary measures $Q_A^1$ and $Q_A^2$, we get that
\begin{equation*}
U_A(w_A + P_A^\nu - F) \ge  U_A(w_A + \nu P_A^1 + (1-\nu)P_A^2 - F)  
\end{equation*}
and by monotonicity of the utility function $U_A$ we obtain from the previous inequality that
\begin{equation*}
P_A^\nu \ge  \nu P_A^1 + (1-\nu)P_A^2 \ ,
\end{equation*}
thus proving concavity of the seller price function $P_A$.

We will now see that $G_{Q_A^1,Q_A^2}(\nu)$ is non-negative for every $Q_1,Q_2\in\Delta^K$ and $\nu\in[0,1]$. We first note that if $P_A^1=P_A^2$, then $G_{Q_A^1,Q_A^2}(\nu)$ is identically zero. It remains to see that $G_{Q_A^1,Q_A^2}(\nu)$ is non-negative provided $P_A^1\ne P_A^2$. Without loss of generality, we assume that $P_A^1>P_A^2$. Rearranging terms, we write $G_{Q_A^1,Q_A^2}(\nu)$ as
\begin{eqnarray}\label{G_1}
G_{Q_A^1,Q_A^2}(\nu) & = & \nu E_{Q_A^1}	\left [U_A(w_A + P_A^1 - F)-U_A(w_A + \nu P_A^1 + (1-\nu)P_A^2 - F)\right ]  \\
&&+ (1-\nu) E_{Q_A^2}	\left [U_A(w_A + P_A^2 - F) - U_A(w_A + \nu P_A^1 + (1-\nu)P_A^2 - F)\right ] \nonumber \ . 
\end{eqnarray}
Since $U_A$ is differentiable, we obtain that there exists $C_1\in(\nu P_A^1 + (1-\nu)P_A^2,P_A^1)$ such that 
\begin{equation}\label{G_2}
U_A(w_A + P_A^1 - F)-U_A(w_A + \nu P_A^1 + (1-\nu)P_A^2 - F) = (1-\nu)(P_A^1-P_A^2)U_A'(w_A + C_1 - F) \ .
\end{equation}
Similarly, there exists $C_2\in(P_A^2,\nu P_A^1 + (1-\nu)P_A^2)$ such that 
\begin{equation}\label{G_3}
U_A(w_A + P_A^2 - F)-U_A(w_A + \nu P_A^1 + (1-\nu)P_A^2 - F) = -\nu(P_A^1-P_A^2)U_A'(w_A + C_2 - F) \ .
\end{equation}
Substituting \eqref{G_2} and \eqref{G_3} in \eqref{G_1}, we get
\begin{eqnarray}\label{G_6}
G_{Q_A^1,Q_A^2}(\nu) &=&  \nu(1-\nu)(P_A^1-P_A^2) \nonumber \\
&&\times\left( E_{Q_A^1}	\left [U_A'(w_A + C_1 - F)\right ] - E_{Q_A^2}\left [U_A'(w_A + C_2 - F)\right ] \right) \ .
\end{eqnarray}
Recalling that $U_A$ is increasing, we obtain that for every $C_1\in(\nu P_A^1 + (1-\nu)P_A^2,P_A^1)$ we have 
\begin{equation}\label{G_4}
E_{Q_A^1}	\left [U_A(w_A + C^1 - F)\right] < E_{Q_A^1}	\left [U_A(w_A + P_A^1 - F)\right] = U_A(w_A) - r_A \ ,
\end{equation}
and for every $C_2\in(P_A^2,\nu P_A^1 + (1-\nu)P_A^2)$ we have 
\begin{equation}\label{G_5}
E_{Q_A^2}	\left [U_A(w_A + C^2 - F)\right] > E_{Q_A^2}	\left [U_A(w_A + P_A^2 - F)\right] = U_A(w_A) - r_A \ . 
\end{equation}
Combining concavity of $U_A$ with \eqref{G_4} and \eqref{G_5}, we obtain the inequality
\begin{equation*}
E_{Q_A^1}	\left [U_A'(w_A + C^1 - F)\right] > E_{Q_A^2}	\left [U_A'(w_A + C^2 - F)\right] \ ,
\end{equation*}
which, combined with \eqref{G_6}, guarantees that $G_{Q_A^1,Q_A^2}(\nu)$ is non-negative.
\end{proof}

\section{ Proof of Lemma \ref{lemma_bounds}  }\label{proof_lemma_bounds}

\begin{proof}
Since $\Omega$ is a finite set and the contingent claim $F$ is non-constant, we have that
\begin{equation*}
\underset{\omega\in\Omega}{\min}\; F[\omega] \le  F[\omega]  \le  \underset{\omega\in\Omega}{\max}\; F[\omega] \ , \quad \omega\in\Omega \ .
\end{equation*}
From the previous inequalities and strict monotonicity of the utility function $U_A$, we obtain that for every $\omega\in\Omega$
\begin{equation*}
U_A\left(w_A+P_A -\underset{\omega\in\Omega}{\max}\; F[\omega]\right) \le  U_A\left(w_A+P_A -F[\omega]\right)  \le  U_A\left(w_A+P_A -\underset{\omega\in\Omega}{\min}\; F[\omega]\right) \ ,
\end{equation*}
and thus
\begin{equation}\label{PB_1}
U_A\left(w_A+P_A -\underset{\omega\in\Omega}{\max}\; F[\omega]\right) \le  E_{Q_A}\left[U\left(w_A+P_A -F\right)\right]  \le  U_A\left(w_A+P_A -\underset{\omega\in\Omega}{\min}\; F[\omega]\right) \ .
\end{equation}
Recalling the definition of agent A reservation price, we get that
\begin{equation}\label{PB_2}
U_A(w_A)-r_A = E_{Q_A}\left[U\left(w_A+P_A -F\right)\right] \ .
\end{equation}
Combining \eqref{PB_1} and \eqref{PB_2}, we obtain
\begin{equation*}
U_A\left(w_A+P_A -\underset{\omega\in\Omega}{\max}\; F[\omega]\right) \le  U_A(w_A)-r_A  \le  U_A\left(w_A+P_A -\underset{\omega\in\Omega}{\min}\; F[\omega]\right) \ .
\end{equation*}
Solving the previous inequalities with respect to $P_A$, we get the first set of inequalities in the statement. Clearly, one can obtain the second set of inequalities from an analogous computation. 

We now prove the second part of the lemma. Using the bounds for the reservation prices $P_A$ and $P_B$ on the first part of the statement, we obtain that
\begin{equation*}
D^- \le P_B-P_A \le D^+ \ ,
\end{equation*}
where
\begin{eqnarray*}
D^+ &=& G(w_A,w_B,r_A,r_B)  + \underset{\omega\in\Omega}{\max}\; F[\omega] - \underset{\omega\in\Omega}{\min}\; F[\omega] \nonumber \\
D^- &=& G(w_A,w_B,r_A,r_B) -\left( \underset{\omega\in\Omega}{\max}\; F[\omega] - \underset{\omega\in\Omega}{\min}\; F[\omega]\right)  
\end{eqnarray*}
and 
\begin{equation*}
G(w_A,w_B,r_A,r_B) = w_B - U_B^{-1}\left(U_B(w_B)-r_B\right) + w_A - U_A^{-1}\left(U_A(w_A)-r_A\right) \ .
\end{equation*}

Since the risk levels $r_A$ and $r_B$ are assumed to be strictly positive and $F$ is non-constant, we trivially get the following inequalities
\begin{eqnarray*}
w_i - U_i^{-1}\left(U_i(w_i)-r_i\right) & \ge & 0 \ , \quad i= A,B \nonumber \\
\underset{\omega\in\Omega}{\max}\; F[\omega] - \underset{\omega\in\Omega}{\min}\; F[\omega] & > & 0   \  .
\end{eqnarray*}
Using the previous inequalities it is clear that $D^+$ is strictly positive and that $D^-$ is negative if and only if 
\begin{equation*}
\underset{\omega\in\Omega}{\max}\; F[\omega] - \underset{\omega\in\Omega}{\min}\; F[\omega] > w_B - U_B^{-1}\left(U_B(w_B)-r_B\right) + w_A - U_A^{-1}\left(U_A(w_A)-r_A\right) \ ,
\end{equation*}
as required.
\end{proof}

\section{Proof of Lemma \ref{DP_Lip}}\label{proof_DP_Lip}

\begin{proof}
We will prove that the partial derivatives of $P_A$ with respect to the components of $Q_A$ are Lipschitz continuous, the remaining proofs being analogous or easier, such as for the case of the partial derivatives with respect to $r_A$ and $r_B$).

By lemma \ref{P_dif} we have that $P_A$ is a continuously differentiable function of $(r_A,Q_A)$ and that its partial derivatives with respect to the components of $Q_A$ are given by
\begin{equation*}
\frac{\partial P_A}{\partial Q_A^k}(r_A,Q_A) =-\frac{U_A(w_A+P_A-F[k])}{E_{Q_A}[U_A'(w_A+P_A-F)]} \ , \quad k=1,\ldots, K  \ .
\end{equation*}

Fix $k\in\{1,\ldots,K\}$ and $r_A\in R_A$, and let $Q_1,Q_2$ be arbitrary elements of $\Delta^K$. We have that
\begin{eqnarray*}
\lefteqn{\left|\frac{\partial P_A}{\partial Q_A^k}(Q_1) - \frac{\partial P_A}{\partial Q_A^k}(Q_2)\right| =}\\ &&\left|\frac{U_A(w_A+P_A(Q_1)-F[k])}{E_{Q_1}[U_A'(w_A+P_A(Q_1)-F)]}-\frac{U_A(w_A+P_A(Q_2)-F[k])}{E_{Q_2}[U_A'(w_A+P_A(Q_2)-F)]} \right| \nonumber \ ,
\end{eqnarray*}
where we have dropped the dependence of the partial derivatives on $r_A$ to simplify notation. Let us introduce the quantities
\begin{eqnarray*}
D &=& U_A(w_A+P_A(Q_1)-F[k])E_{Q_2}[U_A'(w_A+P_A(Q_2)-F)] \nonumber \\
&&-U_A(w_A+P_A(Q_2)-F[k])E_{Q_1}[U_A'(w_A+P_A(Q_1)-F)] \\
N &=& E_{Q_1}[U_A'(w_A+P_A(Q_1)-F)]E_{Q_2}[U_A'(w_A+P_A(Q_2)-F)] 
\end{eqnarray*}
Clearly, we have that
\begin{equation*}
\left|\frac{\partial P_A}{\partial Q_A^k}(Q_1) - \frac{\partial P_A}{\partial Q_A^k}(Q_2)\right| =\frac{|D|}{|N|} \nonumber \ .
\end{equation*}

By lemma \ref{lemma_bounds}, we get that
\begin{equation}\label{DP_Lip_1}
U_A^{-1}\left(U_A(w_A)-r_A\right) - \Delta[F] \le w_A+P_A-F \le U_A^{-1}\left(U_A(w_A)-r_A\right) + \Delta[F]
\end{equation}
where $\Delta[F]$ is defined as
\begin{equation*}
\Delta[F] = \underset{\omega\in\Omega}{\max}\; F[\omega] -\underset{\omega\in\Omega}{\min}\; F[\omega] \ .
\end{equation*}
Since $\Delta[F]$ is such that $0<\Delta[F]<\infty$, $w_A$ is fixed and $R_A$ is bounded, we obtain that $w_A+P_A-F$ is bounded above. As a consequence $U_A'(w_A+P_A-F)$ is bounded away from zero and we obtain that there exists $M_1>0$ such that $|N|>M_1$.

We will now prove that there exists $M_2>0$ such that $D\leq M_2\left\|Q_1-Q_2\right\|$, where $\left\|\cdot\right\|$ denotes the Euclidean norm in $\R^K$. We start by rewriting $D$ as 
\begin{eqnarray*}
D &=& U_A(w_A+P_A(Q_1)-F[k]) \\
&&\times \left(E_{Q_2}[U_A'(w_A+P_A(Q_2)-F)]-E_{Q_1}[U_A'(w_A+P_A(Q_1)-F)]\right) \nonumber \\
&&+\left(U_A(w_A+P_A(Q_1)-F[k])-U_A(w_A+P_A(Q_2)-F[k])\right)  \\
&& \times E_{Q_1}[U_A'(w_A+P_A(Q_1)-F)] \ .
\end{eqnarray*}
Since $U_A(w_A+P_A(Q_1)-F[k])$ is a continuous function of $Q_1\in\Delta^K$ and $\Delta^K$ is a compact set, we obtain that $|U_A(w_A+P_A(Q_1)-F[k])|$ is bounded above. Moreover, using the inequalities \eqref{DP_Lip_1}, we obtain that $|E_{Q_1}[U_A'(w_A+P_A(Q)-F)]|$ is also bounded above. Recall that $U_A(w_A+P_A(Q)-F[k])$ is continuously differentiable with respect to $Q$ to obtain the existence of $C_1>0$ such that 
\begin{equation*}
|U_A(w_A+P_A(Q_1)-F[k])-U_A(w_A+P_A(Q_2)-F[k])| \le C_1\left\|Q_1-Q_2\right\| \ . 
\end{equation*}
We now note that 
\begin{eqnarray*}
\lefteqn{\left|E_{Q_2}[U_A'(w_A+P_A(Q_2)-F)]-E_{Q_1}[U_A'(w_A+P_A(Q_1)-F)]\right| \le} \nonumber \\
&&\left|E_{Q_2}[U_A'(w_A+P_A(Q_2)-F)-U_A'(w_A+P_A(Q_1)-F)]\right|\\
&&+\left|E_{Q_2}[U_A'(w_A+P_A(Q_1)-F)]-E_{Q_1}[U_A'(w_A+P_A(Q_1)-F)]\right|
\end{eqnarray*}
Since $P_A$ is continuously differentiable and we are assuming that $U_A'$ is Lipschitz continuous, we obtain that there exists $C_2>0$ such that
\begin{eqnarray*}
\left|E_{Q_2}[U_A'(w_A+P_A(Q_2)-F)-U_A'(w_A+P_A(Q_1)-F)]\right| \le  C_2\left\|Q_1-Q_2\right\| \ .
\end{eqnarray*}
Moreover, we get that
\begin{eqnarray*}
\lefteqn{\left|E_{Q_2}[U_A'(w_A+P_A(Q_1)-F)]-E_{Q_1}[U_A'(w_A+P_A(Q_1)-F)]\right|}\\
&& = \left|\sum_{k=1}^K (Q_2^K-Q_1^K)U_A'(w_A+P_A(Q_1)-F[k]) \right|  \\
&& \le \sum_{k=1}^K \left|Q_2^K-Q_1^K\right|\left|U_A'(w_A+P_A(Q_1)-F[k])\right|
\end{eqnarray*}
Using the inequalities \eqref{DP_Lip_1} once again, we obtain that $U_A'(w_A+P_A(Q_1)-F)$ is bounded above. Thus, there exists $C_3>0$ such that 
\begin{eqnarray*}
\left|E_{Q_2}[U_A'(w_A+P_A(Q_1)-F)]-E_{Q_1}[U_A'(w_A+P_A(Q_1)-F)]\right| \le C_3\left\|Q_1-Q_2\right\| \ .
\end{eqnarray*}
We conclude that there exists $M_2>0$ such that $|D|\le M_2\left\|Q_1-Q_2\right\|$. C

Combining this inequality with the bound obtained previously for $|N|$, we get the desired result. 
\end{proof}

\section{Proof of Theorem \ref{opt_existence}}\label{proof_opt_existence}

\begin{proof}
Let us fix $\lambda\in(0,1)$, $r_A,r_B\in\R_0^+$, $w_A,w_B\in\R$ and $Q_A^0,Q_B^0\in\Delta^K$ such that $P(r_A,Q_A^0)>P(r_B,Q_B^0)$, and consider the minimization problem \eqref{PRIMAL}. 

Note that $\psi_A(Q_A,Q_A^0)$ is a convex function of $Q_A$ and $\psi_B(Q_B,Q_B^0)$ is a convex function of $Q_B$. Thus, the function to be minimized 
\begin{equation*}
\lambda \psi_A(Q_A,Q_A^0) + (1-\lambda) \psi_B(Q_B,Q_B^0)
\end{equation*} 
being a convex combination of convex functions, is also convex. Using lemma \ref{P_wd}, we get that for fixed levels of risk $r_A$ and $r_B$, the seller price function $P_A$ is a concave function of $Q_A$ and the buyer price function $P_B$ is a convex function of $Q_B$. Thus, we obtain that $P_B(Q_B)-P_A(Q_A)$ is a convex function of $(Q_A,Q_B)\in\Dcal$ and that
\begin{equation*}
\left\{(Q_A,Q_B)\in\Dcal: P_B(Q_B)\ge P_A(Q_A) \right\}
\end{equation*}
is a convex subset of $\Dcal$. From the convexity of the function to be minimized and the convexity of the constraint set, we obtain that the minimization problem \eqref{PRIMAL} has a unique solution which is attained at $(Q_A^*,Q_B^*)\in\Dcal$.

Using lemma \ref{P_dif} we obtain that the price functions $P_A$ and $P_B$ are differentiable with respect to $Q_A$ and $Q_B$, respectively. One can then use the Lagrange multiplier method to solve \eqref{PRIMAL}. Define the Lagrangian function $L_{r_A,r_B,\lambda}:\Dcal\times\R^+\rightarrow\R$ associated with the minimization problem \eqref{PRIMAL} as
\begin{eqnarray*}
L_{r_A,r_B,\lambda}(Q_A,Q_B,\mu) &=&  \lambda \psi_A(Q_A,Q_A^0) + (1-\lambda) \psi_B(Q_B,Q_B^0) \\
&& + \mu(P_A(r_A,Q_A) - P_B(r_B,Q_B)) \ ,
\end{eqnarray*}
where $r_A$ and $r_B$ are fixed and $\mu$ is the Lagrange multiplier associated with the constraint inequality $P_A(r_A,Q_A) \le P_B(r_A,Q_B)$ in the minimization problem under consideration. We then look for points $(Q_A,Q_B,\mu)\in\Dcal\times\R$ satisfying the following $2K+1$ equalities
\begin{eqnarray}\label{Lagrange_mult}
\frac{\partial L_{r_A,r_B,\lambda}}{\partial Q_A^j}(Q_A,Q_B,\mu) &=& 0 \ , \quad j=1,\ldots, K \nonumber \\
\frac{\partial L_{r_A,r_B,\lambda}}{\partial Q_B^j}(Q_A,Q_B,\mu) &=& 0 \ , \quad j=1,\ldots, K \\
\frac{\partial L_{r_A,r_B,\lambda}}{\partial \mu}(Q_A,Q_B,\mu) &=& 0 \ . \nonumber
\end{eqnarray}
Since we have already proved the existence of a solution $(Q_A^*,Q_B^*)\in\Dcal$ for the minimization problem \eqref{PRIMAL}, Kuhn-Tucker theorem implies the existence of $\mu^*\in\R$ such that $(Q_A^*,Q_B^*,\mu^*)\in\Dcal\times\R$ is a solution of \eqref{Lagrange_mult}. In particular, from the last equation in \eqref{Lagrange_mult} we must have that
\begin{equation*}
P_A(Q_A^*)=P_B(Q_B^*) \ .
\end{equation*}
Since $Q_A^0,Q_B^0\in\Delta^K$ are such that $P(r_A,Q_A^0)>P(r_B,Q_B^0)$, then inequality \eqref{dif_ine} holds and thus the existence of a pair $(Q_A^*,Q_B^*)\in\Delta^K\times \Delta^K$ satisfying the previous equality is guaranteed by lemma \ref{lemma_bounds}.

The continuity of the solution of the optimization problem \eqref{PRIMAL} with respect to $\lambda$, $r_A$ and $r_B$ is a consequence of Berge's maximum theorem \cite[Ch. VI, Sec. 3]{Berge}, which guarantees continuity of the minimal functional 
\begin{equation*}
\lambda \psi_A(Q_A^*,Q_A^0) + (1-\lambda) \psi_B(Q_B^*,Q_B^0)
\end{equation*} 
with respect to $r_A$, $r_B$ and $\lambda$ and upper semicontinuity of the correspondence given by
\begin{equation}\label{Q_cor}
(\lambda,r_A,r_B)\mapsto \left(Q_A^*(\lambda,r_A,r_B),Q_B^*(\lambda,r_A,r_B)\right) \ .  
\end{equation}
Since the solution of the problem is unique by theorem \ref{opt_existence}, we get that correspondence \eqref{Q_cor} is single-valued and therefore continuous.

\end{proof}

\section*{Acknowledgments}

We thank the Calouste Gulbenkian Foundation, PRODYN-ESF, POCTI, and POSI by FCT and Minist\'erio da Ci\^encia, Tecnologia e Ensino Superior, CEMAPRE, LIAAD-INESC Porto LA, Centro de Matem\'atica da Universidade do Minho and Centro de Matem\'atica da Universidade do Porto for their financial support. D. Pinheiro's research was supported by FCT - Funda\c{c}\~ao para a Ci\^encia e Tecnologia program ``Ci\^encia 2007''. D. Pinheiro would also like to acknowledge the financial support from ``Programa Gulbenkian de Est\'imulo \`a Investiga\c{c}\~ao 2006'' and FCT - Funda\c{c}\~ao para a Ci\^encia e Tecnologia grant with reference SFRH / BPD / 27151 / 2006. Parts of this work were done during visits of D. Pinheiro to Athens University of Economics and Business (Greece) and Universit\'e Toulouse III - Paul Sabatier (France), who thanks them for their hospitality. S. Xanthopoulos would like to acknowledge that this project is co-funded by the European Social Fund and National Resources - (EPEAEK-II) PYTHAGORAS.

\bibliography{biblio}
\bibliographystyle{plain} 

\end{document}